\documentclass[a4paper,twocolumn,11pt,accepted=2025-10-22]{quantumarticle}
\pdfoutput=1
\usepackage[utf8]{inputenc}
\usepackage[english]{babel}
\usepackage[T1]{fontenc}
\usepackage{amsmath}
\usepackage{hyperref}

\usepackage{tikz}
\usepackage{lipsum}
\usepackage{amsmath,amsfonts,amsthm, mathtools}
\usepackage{float}
\usepackage{dsfont}
\usepackage{csquotes}
\usepackage{tocloft}
\usepackage{titletoc}
\usepackage{amssymb}
\usepackage{verbatim}
\usepackage{appendix}
\usepackage{bm}
\usepackage{multirow}

\definecolor{darkred}  {rgb}{0.5,0,0}
\definecolor{darkblue} {rgb}{0,0,0.5}
\definecolor{darkgreen}{rgb}{0,0.5,0}
\hypersetup{
  colorlinks = true,
  urlcolor  = blue,         
  linkcolor = darkblue,     
  citecolor = darkgreen,    
  filecolor = darkred       
}
\theoremstyle{definition}

\newtheorem{corollary}{Corollary}
\newtheorem{definition}{Definition}
\newtheorem{conjecture}{Conjecture}
\newtheorem{lemma}{Lemma}
\newtheorem{proposition}{Proposition}
\newtheorem{theorem}{Theorem}

\newtheorem{result}{Result}

\newtheorem*{remark}{Remark}

\newcommand{\mbb}{\mathbb}
\newcommand{\mc}{\mathcal}
\newcommand{\msf}{\mathsf}
\newcommand{\mf}{\mathfrak}

\newcommand{\tr}{\textrm{Tr}}
\newcommand{\wt}{\widetilde}

\usepackage{multirow}
\newcommand{\ket}[1]{|#1\rangle}
\newcommand{\bra}[1]{\langle #1|}
\newcommand{\op}[2]{|#1\rangle\langle#2|}

\newcommand{\inr}{\text{inr}}
\newcommand{\rot}{\text{rot}}

\definecolor{cool_green}{rgb}{0.0, 0.5, 0.0}

\newcommand{\yujie}{\color{black}}
\newcommand{\blk}{\color{black}}
\newcommand{\norm}[1]{\left\lVert#1\right\rVert}
\usepackage{blkarray}
\usepackage{graphicx}
\usepackage{dcolumn}
\usepackage{bm}

\begin{document}

\title{Cost of Simulating Entanglement in Steering Scenarios}

\author{Yujie Zhang}
 \affiliation{Department of Physics, University of Illinois at Urbana-Champaign, Urbana, IL 61801, USA}
\orcid{0000-0002-7858-7476}
\email{yujie4physics@gmail.com}
\author{Jiaxuan Zhang}%
 \affiliation{Department of Physics, University of Illinois at Urbana-Champaign, Urbana, IL 61801, USA}
\author{Eric Chitambar}
\email{echitamb@illinois.edu}
\affiliation{
Department of Electrical and Computer Engineering, University of Illinois at Urbana-Champaign, Urbana, IL 61801, USA
}%
\maketitle
\begin{abstract}
Quantum entanglement is a fundamental feature of quantum mechanics, yet certain entangled states that are unsteerable can be classically simulated in steering scenarios, making them unable to exhibit quantum steering. Despite their significance, a systematic comparison of such entangled states has not been explored. In this work, we quantify the resource content of unsteerable quantum states in terms of the amount of shared randomness required to simulate the assemblages they generate in the steering scenario. We rigorously demonstrate that the simulation cost is unbounded even for certain unsteerable two-qubit states. Moreover, the simulation cost of entangled two-qubit states is always strictly larger than that for any separable state. \par 

A significant portion of our results rests on the relationship between the simulation cost of two-qubit Werner states and that of noisy spin measurements. Using noisy spin measurements as our central example, we also investigate the minimum number of outcomes a parent measurement requires to simulate a given set of compatible measurements. Although certain continuous measurement families admit a finite-outcome parent measurement, we identify scenarios where the simulation cost is unbounded. Our results establish previously unknown lower bounds and upper bounds on the shared randomness simulation cost, supported by connections between the simulation cost of noisy spin measurements and various geometric inequalities, including ones from the zonotope approximation problem in Banach space theory.
\end{abstract}

\section{Introduction}
In quantum information theory, entanglement is recognized as a resource because it enables certain dynamical processes (i.e., channels)  that would otherwise be impossible, or at least more costly, to perform.  For instance, shared entanglement can be used in conjunction with some noisy classical channel to increase its zero-error transmission \cite{Cubitt-2010a, Leung-2012a}.  Another dramatic effect is quantum teleportation, in which a classical channel is combined with shared entanglement to generate a perfect quantum channel \cite{Bennett-1993a}.  In these examples, entanglement is used for the purpose of channel coding as it helps transform a noisy channel into a noiseless one.  The converse task, analogous to the ``Reverse Shannon'' problem \cite{Bennett-2002a, Cubitt-2011a}, considers how much noiseless resource is needed to simulate a noisy entanglement-enhanced setup.  Here the goal is not communication but rather simulation.  From this perspective, one can quantify the resource value of an entangled state by how much classical resources is needed to simulate the dynamics that it generates.

Beyond the standard entanglement-assisted paradigm, there are various other ways in which a bipartite entangled state can be used to build a channel \cite{Schmid-2020a}, and each of them will typically have a different simulation cost. In this paper, we focus on the type of channels that emerge in quantum steering setups \cite{Uola2020}.  These are classical-to-classical-quantum (c-to-cq) channels that are generated when Alice measures her half of the state $\rho_{AB}$ in some chosen manner, thereby preparing Bob's system into one of many possible post-measurement states.  From a fundamental perspective, these channels are vital to the study of quantum nonlocality \cite{Wiseman2007}.  But more practically they have important applications in semi-device-independent entanglement witnessing \cite{Cavalcanti-2017a} and operational games \cite{Piani-2015a}.  

\begin{figure}[b]
    \centering
\includegraphics[width=0.32\textwidth]{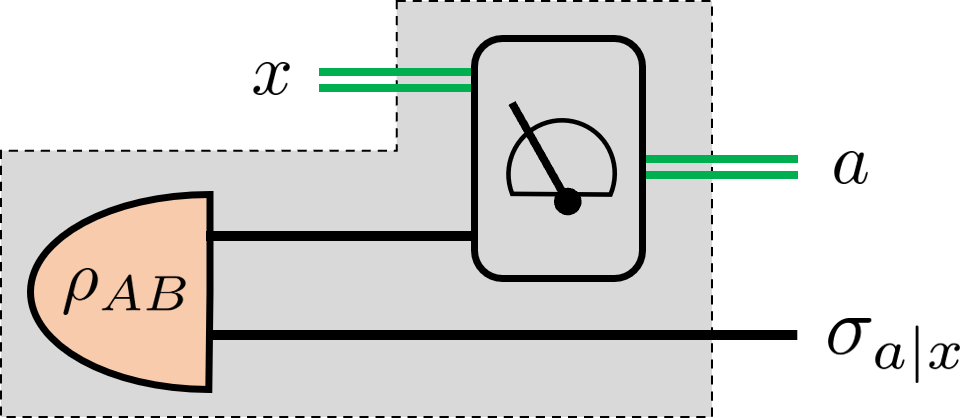}
    \caption{The bipartite state $\rho_{AB}$ induces a c-to-cq channel represented by the assemblage $\mc{A}=\{\sigma_{a|x}\}$.  }
    \label{fig:resource_transformation}
\end{figure}

As shown in Fig. \ref{fig:resource_transformation}, each input $x$ to the channel corresponds to a local measurement choice for Alice, which we collectively represent by a family of positive operator-valued measures (POVMs) $\{M_{a|x}\}_{a,x}$.  She receives a classical output $a$ while Bob receives the (unnormalized) quantum output 
\begin{equation}
\label{Eq:channel-conversion}
\sigma_{a|x}=\tr_A[(M_{a|x}\otimes\mbb{I})\rho_{AB}]
\end{equation}
with probability $p(a|x)=\tr[\sigma_{a|x}]$. It is convenient to denote a c-to-cq channel by the set of labeled quantum outputs $\mc{A}=\{\sigma_{a|x}\}$, called a state assemblage \cite{Uola2020}, and we write $\rho_{AB}\mapsto\mc{A}$ when Eq. \eqref{Eq:channel-conversion} holds.

To quantify the classical simulation cost of entanglement, many previous works adopt the feed-forward classical communication cost (measured in bits) as their figure of merit. \yujie  Notable examples include the trade-off between ebits and cbits in quantum-channel simulation \cite{Bennett2014}, the communication cost of reproducing Bell nonlocal correlations \cite{Toner2003,Renner2023} and the communication cost of reproducing quantum steering \cite{Sainz2016}. However, any entangled state admitting a local hidden variable (LHV) model or a local hidden state (LHS) model can already be simulated with zero forward communication in its corresponding scenario. This causes communication-based metrics to classify such states as ``free", even though Bell-local states or unsteerable states that are entangled remain useful in many quantum applications\cite{Cavalcanti, Piani2009}.  Moreover, implementing an LHV model or an LHS model generally demands a nontrivial amount of shared randomness, whose distribution itself, in fact, must invoke some sort of prior communication. Therefore, a complete quantification of classical simulation cost should account not only for the communication cost but also for the shared randomness cost needed to reproduce the desired correlations.

In this paper, we study the simulation of quantum correlations using only shared randomness—with no forward communication. While shared randomness alone cannot produce Bell nonlocal correlations or demonstrate quantum steering, it can reproduce the measurement statistics of a broad class of entangled states \cite{Hirsch2016,Bowles2015}. There is currently no general method to quantify this share randomness cost, particularly to obtain tight lower bounds on the required amount of shared randomness. In this work, we focus on simulation cost problem in the steering scenario. Although less commonly considered than the Bell scenario, its comparatively simple characterization lets us uncover several key features of the shared randomness simulation cost, including its unboundedness and its operational relevance for separating separable from entangled states.

The benefit of a systematic study of shared randomness simulation cost is twofold. On the one hand, unbounded shared randomness has long been invoked in LHV or LHS models for steering or Bell nonlocality \cite{Werner1989, Wiseman2007} and is often believed to be necessary. Yet, a rigorous proof of the necessity of such unbounded shared randomness has been lacking; our paper provides the first concrete demonstration. On the other hand, our findings provide a systematic way to quantify simulation cost across different entangled states and different sets of compatible measurements. This reveals complexity-based advantages of quantum resources and offers a complementary operational understanding of how quantum resources can be harnessed beyond existing paradigms. For instance, implementing incompatible measurements often requires additional quantum memory \cite{Buscemi-2020a}. Although compatible measurements can always be implemented with classical memory, those with high simulation cost may still serve as valuable resources because their classical simulation requires high-dimensional—or even unbounded—classical memory \cite{zhang2024towards}. \blk

\blk

\section{Preliminary}
For a given quantum state $\rho_{AB}$, we wish to quantify the amount of classical resources needed to simulate each assemblage $\mc{A}$ for which $\rho_{AB}\mapsto\mc{A}$.  For this simulation task, it is natural to quantify the classical resource in terms of how much shared randomness Alice and Bob need to reproduce $\mc{A}$ using a local hidden state (LHS) model.  An LHS model for assemblage $\mc{A}=\{\sigma_{a|x}\}$, should it exist, is given by a shared random variable $\lambda$ with distribution $p(\lambda)>0$, a family of quantum states $\{\rho_{\lambda}\}$ for Bob, and a classical channel $p(a|x,\lambda)$ for Alice such that 
\begin{equation}  
\label{Eq:LHS}
\sigma_{a|x}=\sum_{\lambda\in \Lambda}p(a|x,\lambda)p(\lambda)\rho_{\lambda}\qquad  \sigma_{a|x}\in \mc{A},~\forall a,x.
\end{equation}
{Deploying this model requires Alice and Bob to have shared randomness of size $|\Lambda|\in[1,+\infty]$}, and we define $\gamma(\mc{A}):=\inf |\Lambda|$ as the smallest possible cardinality of $\Lambda$ for which Eq. \eqref{Eq:LHS} holds, and we refer to this as the \textit{simulation cost} of state assemblage $\mc{A}$. {For LHS models using a continuous shared variable, the sum in Eq. \eqref{Eq:LHS} is replaced with an integral, and thus the simulation cost will exist but could be unbounded.}  

If an assemblage $\mc{A}$ does not admit an LHS model, then the resource transformation  $\rho_{AB}\mapsto\mc{A}$ constitutes the celebrated effect of quantum steering \cite{Wiseman2007}, specifically in the direction from Alice to Bob. Although entanglement is a necessary condition for steering, not all entangled states are steerable.  For such states, the next natural question then is to look \textit{beyond steering} and ask how much shared randomness is needed to simulate all the assemblages $\mc{A}$ that the state can generate. For an unsteerable $\rho_{AB}$, we define \begin{equation}
    \gamma(\rho_{AB}):=\sup_{\mc{A}}\gamma(\mc{A})
\end{equation}
as the highest simulation cost to simulate any $\mc{A}$ for which $\rho_{AB}\mapsto\mc{A}$.   

To be fully comprehensive, one should also consider the simulation cost $\gamma(\mc{B})$ for any assemblage $\mc{B}$ generated from $\rho_{AB}$ by locally measuring on Bob's side since quantum steering is asymmetric between the two parties \cite{Bowles2014, Bowles2016}. However, the specific examples we study in this paper involve symmetric quantum states (under the exchange of parties).  Consequently, our results remain general by just considering assemblages generated by local measurements on Alice's side. As an interesting special case that will be discussed in the paper, when Alice and Bob hold a qubit-qubit state, we also define 
\begin{equation}
    \gamma^p(\rho_{AB}):=\sup_{\mc{A}^p}\gamma(\mc{A}^p)
\end{equation}
as the simulation cost of simulating any $\mc{A}^p$ generated by measurements that are confined to the $x$-$z$ plane, i.e., $\{M_{a|x}:=\alpha_{a|x}(\mbb{I}+\hat{n}_{a|x}\cdot\vec{\sigma})\}$ in Eq.~\eqref{Eq:channel-conversion} with Bloch vectors $\{\hat{n}_{a|x}\}$ lying in the same plane \cite{Bavaresco2017}. 
\blk

One important question we consider is whether the simulation cost $\gamma(\rho_{AB})$ satisfies some general dimensionality bound.  In fact, we always have $\gamma(\rho_{AB})\leq (\dim(A))^2(\dim(B))^2$ for every separable state $\rho_{AB}$ due to the existence of minimal separable decompositions \cite{Uhlmann-1998a}. Does this bound also hold for unsteerable entangled states? This question was first raised and left unanswered by Brunner et al. in a study that established an upper bound on the simulation cost by explicitly constructing local hidden state models \cite{Bowles2015}. 

\yujie In this work, we answer the question negatively: the simulation cost is unbounded for some unsteerable entangled states. \blk We do so by deriving lower bounds on simulation cost that have not been previously discussed. Specifically, we prove that $\gamma(\rho_{AB})$ can be unbounded by demonstrating that its lower bound diverges (tends to infinity), even for certain two-qubit states.

\yujie 
To develop the analysis, we leverage the close link between steering and measurement incompatibility \cite{Quintino2014, Uola2014, Brunner2014}. We now briefly review the notion of measurement compatibility.

A family of measurements $\{M_{a|x}\}_{a,x}$ is called compatible or jointly measurable if each measurement in the family can be simulated by a single `parent' POVM $\{\Pi_{\lambda}\}_{\lambda\in \Lambda}$ with index set $\Lambda$ as:
\begin{align}
\label{Eq:compatible_POVM}
    M_{a|x}=\sum_{\lambda\in \Lambda} p(a|x,\lambda)\Pi_{\lambda}\qquad \forall a, x
\end{align}
for some classical channel $p(a|x,\lambda)$ \cite{Heinosaari-2016a, Guerini-2017a}.  In the simple case of quantum observables, compatibility amounts to pairwise commutativity, but the latter is not necessary for general POVMs. \par 

In parallel with the simulation cost of steering assemblages and bipartite quantum states $\gamma(\rho_{AB})$, Skrzypczyk \textit{et al.} introduced the notion of \textit{compatible complexity}, which, for a family of compatible measurements $\mc{M}:=\{M_{a|x}\}$, is the smallest number of outcomes $|\Lambda|$ such that Eq. \eqref{Eq:compatible_POVM} holds \cite{Skrzypczyk2020}. We interpret this as the simulation cost of a set of compatible measurements $\mc{M}$ and denote it by $N(\mc M)$.

In this paper, we first prove that $N(\mathcal{M})$ can be unbounded for a natural class of compatible measurements—namely, the noisy spin measurements introduced below. Using the connection between measurement incompatibility and EPR steering (Propositions~\ref{prop:general CC}–\ref{prop:planar CC}), this yields an unbounded simulation cost $\gamma(\rho_{AB})$ for two-qubit Werner states and, more generally, for a broader class of entangled yet unsteerable states via reductions to Werner states.
\blk

\par 

\section{Simulation Cost of Werner States and Noisy Spin Measurements}
\begin{table}[t]
\small
    \centering
    \begin{tabular}{c|c|c}
    Singlet Weight & SR Cost $\gamma(\omega_r)$ & State type\\
        \hline
$r=0$ & $\gamma=1$ & product state \\
\hline
$0<r\leq\frac{1}{3}$     &   $\gamma=4$  &   separable \\
\hline
\multirow{2}{5em}{$\frac{1}{3}<r\leq\frac{1}{2}$}      &   $\gamma>4\;$ ;    &  entangled and \\
& $\gamma=\Omega\big((\frac{1}{2}-r)^{-\frac{2}{5}}\big)\;$  & unsteerable  \\
     \hline
     $\frac{1}{2}<r<1$ &  --- & steerable
     \end{tabular}
    \caption{Simulation cost of two-qubit Werner states $\omega_r$.}
        \label{tab: summary}
\end{table}
\yujie The bulk of our analysis focuses on:

\noindent (1):the two-qubit family of Werner states.

\begin{equation}
\label{Eq:Werner}
\omega_r=r\op{\Psi^-}{\Psi^-}+(1-r)\frac{\mbb{I}\otimes\mbb{I}}{4},
\end{equation}
where $\ket{\Psi^-}=\sqrt{1/2}(\ket{01}-\ket{10})$ is the singlet state;
\\
\noindent (2): the family of noisy spin measurements (i.e, noisy qubit projective-valued measure):
\begin{equation}
\mc{P}_r=\{\{M_{+|\hat{n}}^{(r)}, M_{-|\hat{n}}^{(r)}\}\}_{\hat{n}},\quad M_{\pm|\hat{n}}^{(r)}=\tfrac{1}{2}(\mbb{I}\pm r\hat{n}\cdot\vec{\sigma})
\label{eq:nspin}
\end{equation}
with $\hat{n}$ being a normalized Bloch vector, and the noise parameter $r\in[0,1]$ effectively shrinks the Bloch sphere down to radius $r$.  We similarly write $\mc{P}_r^p$ for the subset of all noisy spin measurements with $\hat{n}$ confined to the $x$-$z$ plane. 

The Werner states and the noisy spin measurements provide paradigmatic cases for the study of quantum steering, and measurement incompatibility, respectively\cite{Werner1989, Wiseman2007, Uola2014, Brunner2014,Heinosaari-2015a}. It is known that $\mc{P}_r$ is a compatible family if and only if $r\in [0,\frac{1}{2}]$ \cite{Wiseman2007,Heinosaari-2015a}, whereas the compatibility range for $\mc{P}_r^p$ (the planar spin measurements) is $r\in [0,\frac{2}{\pi}]$~\cite{Uola2016} (see also Appendix~\ref{appendixE} for explicit construction). Similarly, it was shown recently that $\omega_r$ is unsteerable if $r\in [0,\frac{1}{2}]$\cite{Zhang2024,Renner2024} even when one considers all POVMs, and is unsteerable when $r\in [0,\frac{2}{\pi}]$ under planar POVMs (see also Proposition~\ref{prop:planar CC}).

\begin{figure*}[t]
        \centering
\includegraphics[width=0.7\textwidth]{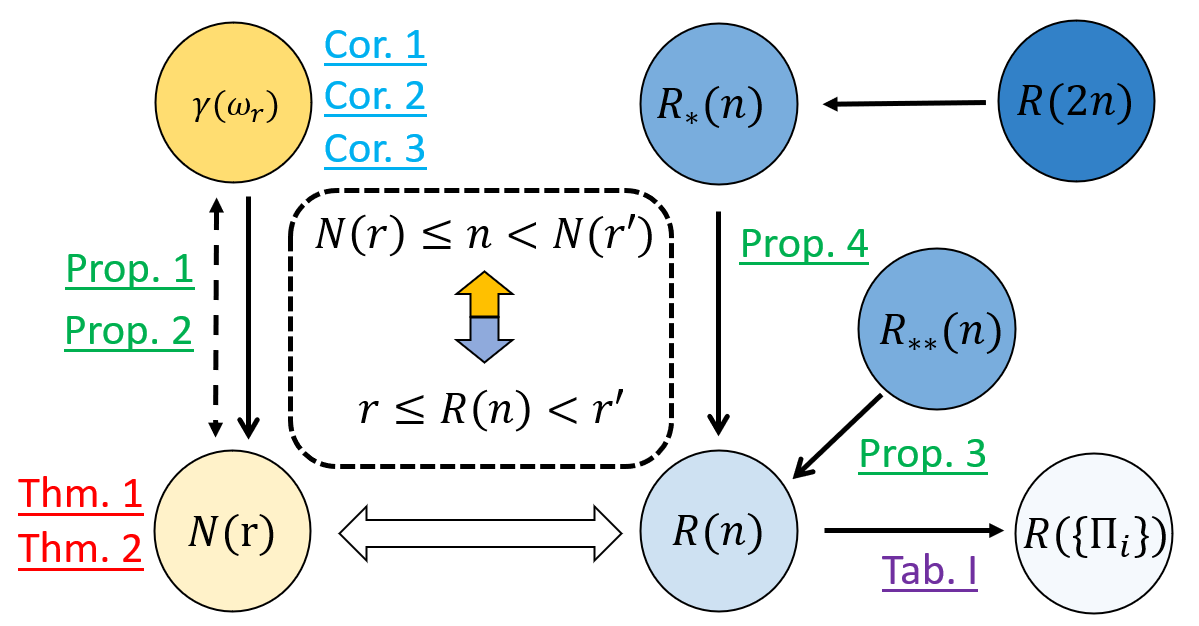}
    \caption{Schematic showing the connections between different quantities in this paper: (a) Left: Relation between simulation cost of Werner states $\gamma(\omega_r)$ and simulation cost (compatible complexity) $N(r)$ for the set of all noisy spin measurements 
    ; (b) Right: Compatible radius for noisy spin measurements $R(n)$ and its derivatives, i.e., various lower and upper bounds; (c) The left and right parts are connected by the relation in the middle box (see Eq.~\ref{Eq:compatibility-radius-complexity}); (d) $\rightarrow$(dashed): Upper bound from left to right (in special cases). }
    \label{fig:schematic}
\end{figure*}

In this paper, however, rather than asking whether two-qubit Werner state or the set of noisy spin measurements is unsteerable or compatible, we focus on their simulation cost within their corresponding parameter regimes. We write $\gamma(\omega_r)$ and $\gamma^p(\omega_r)$ for the simulation cost of the Werner state under all POVMs and under planar POVMs, respectively.  This simulation cost is a special case of those for general bipartite state $\rho_{AB}$, more formally:
\begin{definition}
The simulation cost $\gamma(\omega_r)$ is the smallest $n$ such that every assemblage $\mc{A}:=\{\sigma_{a|x}\}$ generated from $\omega_r$ by Eq.~\ref{Eq:channel-conversion} admits an LHS model with at most $n$ hidden states.
\begin{equation}
\gamma(\omega_r):=\inf\{n: \text{cost for simulating every $\mc{A}$ of $\omega_r$}\},
\end{equation}
\end{definition}
Likewise, we write $N(r)$ and $N^p(r)$ for the simulation costs (compatible complexities) of the measurement families $\mc{P}_r$ and $\mc{P}^p_r$ within their compatibility regions, this simulation cost is a special case of those for general set of measurements $\mc M=\{M_{a|x}\}$, more formally:
\begin{definition}
The simulation cost $N(r)$ of the noisy spin measurements $\mc{P}_r$ is the smallest $n$ for which there exists an $n$-outcome parent POVM $\{\Pi_\lambda\}$ such that every $\{M^{(r)}_{\pm|\hat n}\}\in\mc{P}_r$ can be simulated by Eq.~\eqref{Eq:compatible_POVM}.
\begin{equation}
N(r):=\inf\{n: \text{cost for simulating $\mc{P}_r$}\}
\end{equation}
\end{definition}

Our restriction to this simple family of states and sets of measurements may seem like a limitation.  However, on the contrary, we show that many interesting features of simulation cost can already be demonstrated with these examples, and for the family of Werner states, they are summarized in Table \ref{tab: summary}.  Furthermore, as we explain in this paper, the lower bounds we prove on $\gamma(\omega_r)$ can be translated into lower bounds on $\gamma(\rho_{AB})$ for \textit{any} unsteerable state $\rho_{AB}$.  

Another reason for focusing on Werner states and noisy spin measurements is their deep connection: the unsteerability of Werner states and the compatibility of noisy spin measurements are two sides of the same phenomenon, a link studied extensively in \cite{Uola2014, Quintino2014}. Here we extend this connection to the level of simulation cost, which partially enables bidirectional translation of results. \blk
\begin{lemma}\label{lemma1}
$\mc{P}_{r}$ can be simulated with an $n$-element POVM if and only if the simulation cost of Werner states $\omega_r$ under \textit{projective measurements} equals $n$.
$\gamma_{\rm PM}(\omega_r) = N(r)$.
\end{lemma}
\yujie where $\gamma_{\rm PM}(\omega_r)$ is the simulation cost of $\omega_r$ when only projective-valued measures are considered. \blk

The proof of Lemma~\ref{lemma1} is detailed in Appendix~\ref{appendixF}, this connection allows us to directly write the following propositions. 
\blk
\begin{proposition}
$\gamma(\omega_r) \geq N(r)$.
    \label{prop:general CC}
\end{proposition}
\noindent This inequality follows immediately from the definitions: $\gamma(\omega_r)$ is the simulation cost of the Werner state under arbitrary POVMs on Alice’s side, which include projective measurements as a special case. However, strong numerical evidence suggests that simulating all noisy POVMs is strictly more costly than simulating noisy spin PVMs. We discuss this in Appendix~\ref{appendixG} and summarize it in the following conjecture.
\begin{conjecture}
The set of noisy qubit POVMs and noisy qubit PVMs (spin measurements) have different simulation costs. 
\label{conj:POVM vs PVM}
\end{conjecture}
\noindent Thus, it appears the above inequality might not be tight in general.\par 

Nevertheless, in a special case, we further show in the appendix~\ref{appendixG} that this inequality can be saturated for planar measurements.
\begin{proposition}
    $\gamma^p(\omega_r)=N^p(r)$ for $r\in[0,\frac{2}{\pi}]$.
    \label{prop:planar CC}
\end{proposition}
In what follows, we begin by studying the simulation cost $N(r)$ of noisy spin measurements. The connections made in Proposition \ref{prop:general CC} and Proposition \ref{prop:planar CC} will then help us establish results for the simulation cost $\gamma(\omega_r)$ of two-qubit Werner states.  More general results on the simulation cost of arbitrary bipartite states will be discussed in the last part of the paper.

\section{Compatibility Radius and Its Geometric Characterization}
Propositions \ref{prop:general CC} and \ref{prop:planar CC} directly translate the problem of simulating $\omega_r$ into the problem of simulating noisy spin measurements $\mc{P}_r$, and accordingly the simulation cost $\gamma(\omega_r)$ to the simulation cost $N(r)$.  \yujie However,  Because $N(r)$ is an integer-valued step function of $r$, it is awkward to analyze directly. We therefore invert the viewpoint: instead of fixing the visibility $r$ and asking for the simulation cost $N(r)$ for $\mc{P}_r$. Here we fix $n$ and define $R(n)$ as the largest $r$ for which there exists an $n$-outcome POVM that can simulate $\mc{P}_r$, i.e., 
\begin{definition}
The \textit{compatibility radius} $R(n)$ is the largest visibility $r$ such that the entire family of noisy spin measurements $\mc{P}_r$ can be simulated via Eq.~\eqref{Eq:compatible_POVM} by some $n$-outcome POVM $\{\Pi_{\lambda}\}_{\lambda=1}^n$.
\label{def:cradius}
\begin{align}
&N(r)=\inf\{n: \text{cost for simulating $\mc{P}_r$}\}  \\
\Rightarrow &R(n)=\sup\{r: \text{$\mc{P}_r$ is simulable with cost $n$}\}. \notag 
\end{align}
\end{definition}
\noindent That is, 
\begin{align}
N(r):&=\inf\{n: r\le R(n)\} \\
\Rightarrow  R(n):&=\sup\{r: N(r)\le n\}, \notag
\end{align}
thus, for any $r<r'\in[0,\frac{1}{2}]$ and integer $n$, we have the relationship:
 \begin{equation}
 \label{Eq:compatibility-radius-complexity}
r \leq R(n)< r' \;\;\Leftrightarrow\;\; N(r)\leq  n < N(r').
 \end{equation}
The same applies in the planar case with $N^p(r)$ and $R^p(n)$, i.e.,  $r \leq R^p(n)< r' \;\Leftrightarrow\; N^p(r)\leq  n < N^p(r')$\footnote{The equality cannot, in general, be taken on the right-hand side because $R(n)$ is defined as a supremum and is only right-continuous by definition.}.

\par  From Propositions \ref{prop:general CC} and \ref{prop:planar CC}, bounds on $\gamma(\omega_r)$ and $\gamma^p(\omega_r)$ can therefore be obtained by studying $R(n)$ and $R^p(n)$. This inversion allows us study each discrete $n$ directly, and to characterize $R(n)$ and $R^p(n)$ both at finite $n$ and in the asymptotic regime with $n\rightarrow \infty$;  the detailed relation among all these quantitites are summarized in Fig.~\ref{fig:resource_transformation}.\blk

\begin{table}[h]
\small
    \centering
    \begin{tabular}{|c|c|}
    \hline
    \textbf{symbol} & \textbf{Explanation} \\
    \hline
    $\{\Pi_{\lambda}\}_{\lambda\in\Lambda}$     &  POVM with effects $\{\Pi_{\lambda}\}$\\
    \hline
    $\text{sym}\{\Pi_{\lambda}\}_{\lambda\in\Lambda}$     &  Symmetric extended POVM\\
    \hline
    $\Lambda$     & Index set of outcomes\\
         \hline
    $\{M_{a|x}\}_{a,x}$     & Set of Children POVMs \\
        \hline
    $\gamma(\omega_r)$     & Simulation cost of Werner states \\
    \hline
    $N(r)$     & Simulation cost of\\
      & noisy spin measurements \\ 
    \hline  
    $\mathfrak{m}_{\{ \Pi_{\lambda}\}}$     & Compatible region of $\{ \Pi_{\lambda}\} $ \\
         \hline
    $\mathfrak{m}^*_{\{ \Pi_{\lambda}\}}$     & Compatible region of $\text{sym}\{ \Pi_{\lambda}\} $ \\
    \hline    $R(\{ \Pi_{\lambda}\})$     & Compatible radius for $\{ \Pi_{\lambda}\}$\\
    \hline
    $R(n)$  & $R(\{ \Pi_{\lambda}\})$ for all $n$-outcome POVM \\
    \hline
    $R_*(n)$  & Symmetric upper bound of $R(n)$    \\
    \hline    
    $R_{**}(n)$  & Simplex upper bound of $R(n)$ \\
    \hline       
    $X^p$ & Quantity for the planar case\\
    \hline
    \end{tabular}
    \caption{Notation used throughout the paper}
    \label{tab:my_label}
\end{table}

We note that $R(n)$ and $R^p(n)$ are also of great interest on their own. Previously, studies on $R(n)$ and $R^p(n)$ have only been implicitly made in \cite{Bowles2015, Hirsch2016, Cavalcanti2016} by constructing finite-share-randomness LHS or LHV models with fixed outcome number $n$. Instead of optimizing over all $n$-outcome POVMs, these Such constructions typically use symmetric parent POVMs $n$‐outcome POVMs $\{\Pi_\lambda\}$, which yield only lower bounds on $R(n)$ for specific finite $n$ and offering neither extension to arbitrary $n$ nor meaningful upper bounds. In contrast, here we present a thorough investigation of $R(n)$, deriving both lower and upper bounds in the finite and asymptotic regimes by exploring different geometric characterizations and related inequalities.\par

\yujie 
\par  To study $R(n)$, we first fix a parent POVM $\{\Pi_{\lambda}\}$ and define its \textit{compatibility radius} as the largest visibility $r$ for which it simulates all noisy spin measurements $\mathcal P_r$:
\begin{equation}
\label{eq:com-radius}
    R(\{\Pi_{\lambda}\})=\sup\bigl\{r:\mc{P}_r\text{ is simulable by }\{\Pi_{\lambda}\}\bigr\}.
\end{equation}
The $n$-outcome compatible radius $R(n)$ in definition~\ref{def:cradius} is then obtained by optimizing over all $n$-outcomes POVMs: 
\begin{equation}
    R(n)=\sup_{\{ \Pi_{\lambda}\}}R(\{ \Pi_{\lambda}\}).
\end{equation}
\par 

We now characterize $R(\{\Pi_{\lambda}\})$ and $R(n)$ in qubit case, where the geometry is especially useful.

In qubit case, each POVM effect can be written in the Pauli basis as:
\begin{equation}
 \Pi_{\lambda}=\alpha_{\lambda} (\mbb{I}+\eta_{\lambda} \hat{n}_{\lambda}\cdot\vec{\sigma}),
\label{eq:Pauli-basis}
\end{equation}
with $0\le \eta_{\lambda} ,\alpha_{\lambda} \leq 1$, and normalization constraints $\sum_{\lambda}\alpha_{\lambda} =1$ and $\sum_{\lambda}\alpha_{\lambda} \eta_{\lambda} \hat{n}_{\lambda}=\vec{0}$\footnote{ In the appendix~\ref{appendixC}, we further show that it is sufficient to consider $\eta_{\lambda} =1$ when optimizing over all $R(\{\Pi_{{\lambda}}\})$}. 

\yujie Then, a noisy spin measurement $\{M_{\pm|\hat{n}}^{(r)}=\frac{1}{2}(\mbb{I}\pm r\hat{n}\cdot\vec{\sigma})\}$ is simulated by the parent POVM $\{\Pi_{\lambda}=\alpha_{\lambda} (\mbb{I}+\eta_{\lambda} \hat{n}_{\lambda}\cdot\vec{\sigma})\}$ precisely when there exist post-processing coefficients $q_{\lambda}\in[0,1]$ such that one of its effect $ M_{+|\hat{n}}^{(r)}=\sum_{\lambda} q_{\lambda}\Pi_{\lambda}$. This yields the natural geometric conditions:
\begin{subequations}
\label{eq:geoview}
    \begin{align}
 r\hat{n}&=2\sum_{\lambda} q_{\lambda}\alpha_{\lambda}\eta_{\lambda}\hat{n}_{\lambda}  \label{eq:geoview1}\\
\frac{1}{2}&=\sum_{\lambda} q_{\lambda}\alpha_{\lambda}\label{eq:geoview2}
\end{align}
\end{subequations}
These relations underlie our geometric characterizations of $R(\{\Pi_{\lambda}\})$, specifically, we will have: 
\begin{align}
    R(\{ \Pi_{\lambda}\}) :&= \sup \bigl\{r: B(0,r) \subseteq \mf{m}_{\{ \Pi_{\lambda}\}} \bigr\}\notag \\
    &=\text{inr}(\mf{m}_{\{ \Pi_{\lambda}\}}), \label{eq: R-geo}
\end{align}
where $B(0,r) = \{x : \|x\| \le r\}$, $\text{inr}(S)$ denotes inradius, and
\begin{align}
    \mf{m}_{\{ \Pi_{\lambda}\}}=\{2\textstyle\sum_{\lambda}q_{\lambda}\alpha_{\lambda} \eta_{\lambda}\hat{n}_{\lambda}\;&\big|\; 0\le q_{\lambda}\le 1, \notag \\
    &\;\sum_{\lambda}q_{\lambda}\alpha_{\lambda} =\tfrac{1}{2}\}. 
    \label{eq:c-zonotope}
\end{align}
is a convex `constrained zonotope'\footnote{As discussed in Appendix~\ref{appendixA}, the set $\mf{m}_{\{\Pi{\lambda}\}}$ is precisely the region of Bloch vectors corresponding to all unbiased two-outcome measurements simulable by $\{\Pi_{\lambda}\}$ (here “unbiased’’ means $\tr[M_1]=\tr[M_2]$ for a measurement ${M_1,M_2}$). Thereforefore,  we will sometimes simply refer to $\mf{m}_{\{\Pi{\lambda}\}}$ as the compatibility region. } obtained by intersecting the zonotope generated by $\{2\alpha_{\lambda} \eta_{\lambda}\hat{n}_{\lambda}\}$ \cite{McMullen1971, ziegler2012}, namely:
\begin{equation}
   \mf{m}^*_{\{ \Pi_{\lambda}\}}=\left\{2\textstyle\sum_{\lambda}q_{\lambda}\alpha_{\lambda} \eta_{\lambda}\hat{n}_{\lambda}\;\big|\; 0\le q_{\lambda}\le 1\right\} 
\label{eq:outerzonotope}
\end{equation}
with the affine hyperplane $\sum_{\lambda}q_{\lambda}\alpha_\lambda=\frac{1}{2}$ from Eq.~\eqref{eq:geoview2}. Here the set $ \mf{m}^*_{\{ \Pi_{\lambda}\}}$ is an important outer approximation of $\mf{m}_{\{ \Pi_{\lambda}\}}$, which coincides with the ``constrained zonotope" associated with the $2n$-outcome POVM $\textbf{sym}\{ \Pi_{\lambda}\}_{\lambda}$:
\begin{equation}
  \textbf{sym}\{ \Pi_{\lambda}\}_{\lambda}:=\{\tfrac{\alpha_{\lambda} }{2}(\mbb{I}+\eta_{\lambda}\hat{n}_{\lambda}\cdot\vec{\sigma}), \tfrac{\alpha_{\lambda} }{2}(\mbb{I}-\eta_{\lambda}\hat{n}_{\lambda}\cdot\vec{\sigma})\}_{\lambda}.   \notag 
\end{equation} 
 as discussed in the appendix~\ref{appendixA}. 

By definition, $R(n)$ is obtained by taking the supremum $\inr(\mf{m}_{\{ \Pi_{\lambda}\}})$. But optimizing this quantity directly is challenging. In contrast, the outer approximation $\mf{m}^*_{\{ \Pi_{\lambda}\}}$-a zonotope- admits a simpler characterization and is related to $R(n)$ via the following chain of inequalities:
\begin{align}
\label{Eq:radius-n-ineq}
    R(n)&=\sup_{\{ \Pi_{\lambda}\}_{i=1}^n}\inr(\mf{m}_{\{ \Pi_{\lambda}\}})\\ & \leq\sup_{\{ \Pi_{\lambda}\}_{i=1}^n}\inr(\mf{m}^*_{\{ \Pi_{\lambda}\}})=:R_*(n)\leq R(2n). \notag
\end{align}
\yujie The first inequality holds because $\mf{m}^*_{\{ \Pi_{\lambda}\}}$ is an outer approximation of $\mf{m}_{\{ \Pi_{\lambda}\}}$,  while the last inequality follows since $R_*(n)$ is the greatest compatible radius achievable by symmetric $2n$-outcome measurement, which form a strict subset of all $2n$-outcome measurement. 

In what follows, many of our results hinge on the characterization of the inradius of the zonotope $\inr(\mf{m}^*_{\{ \Pi_{\lambda}\}})$. 
These analysis yield new upper bounds on $R(n)$ via $R(n)\le R_{*}(n)$.
Moreover, since $R_{*}(n)$ lower-bounds $R(2n)$, the same analysis provides lower bounds on $R(2n)$, in particular, it allows us to obtain results in the asymptotic regime with $n\to\infty$, which have never been explored before.

\section{Lower Bound on Simulation Cost}  
\yujie In the previous section, we introduced the compatibility radius $R(n)$ as a dual perspective on simulation cost $N(r)$ for noisy spin measurements $\mc{P}_r$. And proposition~\ref{prop:general CC} allows us to connect lower bound of $N(r)$ with lower bound on $\gamma(\omega_r)$. We now show that these links allow us to obtain lower bounds on simulation costs from upper bounds on  $R(n)$ and to reinterpret the problem geometrically in terms of various geometric inequalities. 
\blk 

We begin with the simplest cases $n=3$ and $n=4$, where tight upper bounds follow from elementary geometric inequalities.
\begin{proposition}
\label{Prop:Main}
 $R^p(3)=\tfrac{1}{2}$ 
 and $R(4)=\tfrac{1}{3}$.
\end{proposition}
\begin{proof}
The proof relies on the well-known circumradius–inradius inequality, which states that the inradius of an $n$-simplex (e.g., a triangle is a 2-simplex) is at least $n$ times smaller than its circumradius \cite{Murray1979}. 
\par 
To apply this, consider any three-outcome POVM $\{\Pi_{\lambda}\}_{\lambda=1}^3$. Observe that the set $\mf{m}_{\{ \Pi_{\lambda}\}}$ is contained within the $2$-simplex:
\begin{align}
    \mf{m}_{\{ \Pi_{\lambda}\}}\subseteq \mf{m}^{**}_{\{ \Pi_{\lambda}\}}=\{\sum_{\lambda=1}^3p_{\lambda}\hat{n}_{\lambda}|\sum_{\lambda=1}^3 p_{\lambda}=1, p_{\lambda}\ge 0\}
\end{align}
where $\mf{m}^{**}_{\{ \Pi_{\lambda}\}}$ is simply obtained from Eq.~\ref{eq:c-zonotope} by setting $p_{\lambda}=2q_{\lambda}\eta_{\lambda}\alpha_{\lambda}$ and dropping the constraints $q_{\lambda}\le 1$.\blk),

Since $\{\hat{n}_{\lambda}\}$ are unit vectors,  the circumradius of $\mf{m}^{**}_{\{ \Pi_{\lambda}\}}$ is at most one. \yujie The circumradius–inradius inequality \cite{Murray1979} (outlined in Appendix~\ref{appendixC}) then implies: \blk
\begin{align}
    R^p(3)\leq \sup_{\{ \Pi_{\lambda}\}_{i=1}^3}\!\inr(\mf{m}^{**}_{\{ \Pi_{\lambda}\}})=: R^p_{**}(3)\leq \frac{1}{2}.\notag
\end{align}
An analogous argument for four-outcome POVMs (\yujie corresponds to a 3-simplex \blk) on the full Bloch sphere follows directly, which implies that $R(4)\leq \frac{1}{3}$. \par 

Combining these with the lower bounds for $R(n)$ and $R^p(n)$ in the next section (see Table \ref{tab: symmetric POVM}) establishes the proposition. 
\end{proof}

\begin{corollary}
For any $r>\frac{1}{3}$, we have $\gamma(\omega_r)> 4$. Thus, the simulation cost of any entangled two-qubit Werner state (with $r>\frac{1}{3}$\cite{Werner1989}) is strictly greater than that of any separable two-qubit Werner state.
\label{Cor: Main}
\end{corollary}
Unfortunately, the outer approximation $R_{**}(n)$ is too loose for arbitrary $n$. In the following, we seek tighter outer approximations of $\mathfrak{m}_{\{\Pi_{\lambda}\}}$ and better geometric inequalities.

\yujie We begin again with the planar case. Using an isoperimetric inequality for polygons, we obtain a sharper upper bound on $R^p(n)$ \blk:
\begin{proposition}
For $n$-outcome planar measurement $R^p(n)\le R_*^p(n)=\frac{1}{n}\cot (\frac{\pi}{2n})$. 
\end{proposition}
\begin{proof}
Let $\{ \Pi_{\lambda}\}_{{\lambda}=1}^n$ be an arbitrary $n$-outcome planar POVM.  According to Eq.~\eqref{Eq:radius-n-ineq}, its compatible radius will be upper bounded by the inscribed radius \yujie of $\msf{m}^*_{\{ \Pi_{\lambda}\}}$ in Eq.~\eqref{eq:outerzonotope}\blk, i.e.,
\begin{equation}
  R^p(n) =\sup_{\{ \Pi_{\lambda}\}}\textbf{inr}({\msf{m}^*_{\{ \Pi_{\lambda}\}}})\le  \sup_{\{2\alpha_{\lambda} \hat{n}_{\lambda}\}}\frac{2A}{L}, 
\end{equation}
where  $\msf{m}^*_{\{ \Pi_{\lambda}}\}$ is a $2n$-sided zonogon generated by $\{2\alpha_{\lambda}\eta_{\lambda}\hat{n}_{\lambda}\}$, with area $A$ and perimeter $L=4\sum_{\lambda}\eta_{\lambda}\alpha_{\lambda}\le 4$ (note that the polygon consists of $n$ pairs of sides with length equals to its generators). Here the last inequalities follows simply by triangulating the convex polygon from its incenter, which relates the area $A$, perimeter $L$ and inradius of the polygon. 

\par The isoperimetric inequality of a such a $2n$-sided polygon \cite{blaasjo2005} stipulates that 
\begin{equation}
    A\le \frac{L^2}{8n\tan(\frac{\pi}{2n})}.
\end{equation}
Therefore, we have the upper bound
\begin{equation}
R^p(n)\le \frac{1}{n}\cot (\frac{\pi}{2n})<\frac{2}{\pi}
\end{equation} 
\end{proof}
\noindent Expanding $x\cot x$ about $x=0$ using the Taylor expansion, we obtain 
\begin{align}
R^p(n)\le \frac{1}{n}\cot (\frac{\pi}{2n})<\frac{2}{\pi}\left(1-\frac{1}{3}(\frac{\pi}{2n})^2\right)    
\end{align}
Combining this bound with Eq. \eqref{Eq:compatibility-radius-complexity} and Proposition \ref{prop:planar CC} then yields Corollary \ref{cor:lower planer}.
\begin{corollary}
\label{cor:lower planer}
For any $r\in[0,\frac{2}{\pi}]$, 
\begin{equation}
\gamma^p(\omega_r)=N^p(r)> \sqrt{\frac{\pi}{6}}\left(\frac{2}{\pi}-r\right)^{-\frac{1}{2}}.
\end{equation}
\end{corollary}

\yujie We now move beyond the planar case to the full Bloch sphere. Here we rely on stronger isoperimetric inequalities for zonotopes. \blk The culminating result is stated in the following theorem and discussed in detail in the appendix~\ref{appendixC}. 
\yujie
\begin{theorem} 
\label{thm:SC lowerbound}
For any $r\in[0,\frac{1}{2}]$, 
\begin{align}
&R(n)< \frac{1}{2}-cn^{-\frac{5}{2}}\\
&\gamma(\omega_r)\ge N(r) >c'\left(\frac{1}{2}-r\right)^{-\frac{2}{5}}
\end{align}
for some positive constant $c$ and $c'=c^{\frac{2}{5}}$.
\end{theorem}
\blk
\noindent Both Corollary \ref{cor:lower planer} and Theorem \ref{thm:SC lowerbound} establish our primary claim that the simulation cost of an unsteerable entangled state diverges at $r=\tfrac{1}{2}$ (resp. $\frac{2}{\pi}$) when considering general measurements (resp. planar measurements). \yujie We also note that the constants above trace back to those appearing in the original zonotope isoperimetric inequalities \cite{Bourgain1988, Bourgain1993}; see Appendix~\ref{appendixC} for details. \blk


\section{Upper bound on the simulation cost} 
\yujie We now turn to upper bounds on the simulation cost of unsteerable Werner states $\gamma(\omega_r)$ and noisy spin measurements $N(r)$.  We note, however, due to the inequality $\gamma(\omega_r)\ge N(r)$ in propositions \ref{prop:general CC}, upper bounds on $N(r)$ do not translate into upper bounds on $\gamma(\omega_r)$ (except the planar case).  Nevertheless, the simulation cost results $N(r)$ on noisy spin measurements remain valid on its own.\blk

As before, we use Eq. \eqref{Eq:compatibility-radius-complexity} and turn the problem of upper bounding $N(r)$ into lower bounding $R(n)$.  The latter is further lower bounded by $R(\{ \Pi_{\lambda}\})$ for any fixed choice of POVM $\{ \Pi_{\lambda}\}$.  The same reasoning holds for the planar case.

For qubits, the compatible radius  $R(\{\Pi_{\lambda}\})$ (with $\{\Pi_{\lambda}\}$ written in Pauli basis as in Eq.~\eqref{eq:Pauli-basis})  can actually be efficiently computed using the following optimization-based criterion:

\begin{proposition}
For a given POVM ${\Pi_{\lambda}}$, the compatible radius $R(\{\Pi_{\lambda}\})$ is given by: \begin{equation} R(\{\Pi_{\lambda}\})= \inf_{\substack{|\vec{c}|=1\\ -1\le c_0\le 1}}\sum_{\lambda}|\alpha_{\lambda} (c_0+\eta_{\lambda} \vec{c}\cdot\hat{n}_{\lambda})|. \label{eq:optPi} \end{equation} 
\label{prop:criteria}
\noindent Thus, when taking the supremum over all $n$-outcome POVMs ${\Pi_{\lambda}}$, the compatibility radius $R(n)$ can be expressed as: \begin{equation} R(n)=\sup_{\substack{{\{\Pi_{\lambda}}}\}} \inf_{\substack{|\vec{c}|=1\\-1\le c_0\le 1}}\sum_{\lambda}|\alpha_{\lambda} (c_0+\eta_{\lambda} \vec{c}\cdot\hat{n}_{\lambda})|, \label{eq:optR(n)}\end{equation} 
\end{proposition}
\yujie 
However, in general, the optimization problem is difficult to solve, and the POVMs that maximize the compatibility radius are typically asymmetric (see Table~\ref{tab: symmetric POVM} for a brief overview). Nevertheless, the following corollary provides a useful simplification:
\begin{corollary}
    $R(n)$ is maximized by rank-1 POVM $\{\Pi_n\}$, i.e., $\eta_{\lambda} =1$ for $\lambda\in [n]$. 
\label{coro:criteria}
\end{corollary}

Proposition~\ref{prop:criteria} and Corollary~\ref{coro:criteria} are proved in detail in Appendix~\ref{appendixD}; here we summarize its implications.  Equation~\eqref{eq:optPi} shows that for any qubit POVM $\Pi_\lambda=\alpha_\lambda(\mbb{I}+\eta_\lambda\hat n_\lambda\cdot\vec\sigma)\}$, we can explicitly compute its compatible radius $R(\{\Pi_\lambda\})$ via a concrete minimization that depends only on the parameters $\alpha_{\lambda},\eta_{\lambda}$ and $\hat{n}_{\lambda}$ of $\{\Pi_\lambda\}$. Intuitively, this minimization arises by recasting the problem of finding  $R(\{\Pi_\lambda\})$ into a inclusion problem in Eq~\eqref{eq: R-geo}.    Equation~\eqref{eq:optR(n)} then takes the supremum of these individual compatible radius over all $n$-outcome POVMs.  \blk

\par 

\subsection{Qubit planar measurements}
In general, to obtain the exact value of $R(n)$, one must optimize over all $n$-outcome measurements $\{\Pi_{\lambda}\}_{\lambda=1}^n$. However, some well-chosen $n$-outcome measurements can typically provide a relatively tight lower bound on $R(n)$. \par

Intuitively, POVMs $\{\Pi_{\lambda}\}_{\lambda=1}^n$ with a large compatible radius are those that are "symmetric". For example, consider qubit planar measurements. We first examine the rotationally symmetric planar POVM $\{\Pi^{\rot}_{\lambda}=\frac{1}{n}(\mbb{I}+\hat{n}_{\lambda}\cdot\vec{\sigma})\}_{\lambda=1}^n$, where $\hat{n}_{\lambda}=\left(\cos(\frac{2\pi \lambda}{n}),0,\sin(\frac{2\pi \lambda}{n})\right)^T$ \cite{Bavaresco2017}. In the appendix~\ref{appendixE}, we analytically determine their compatibility radius to be
\begin{equation}
R^p(\{\Pi^{\rot}_{\lambda}\})=\begin{cases}
\frac{1}{n}\cot(\frac{\pi}{2n})\cos(\frac{\pi}{2n})~~&\text{if $n$ is odd}\\
\frac{2}{n}\cot(\frac{\pi}{n}) &\text{if $n$ is even}
\end{cases}.
\label{eq: planar sp}
\end{equation} 
Similar to the proof of Corollary \ref{cor:lower planer}, \yujie one can check the asymptotic limit of the above compatibility radius via Taylor expansion 
\begin{align}
R^p(\{\Pi^{\rot}_{\lambda}\})=\begin{cases}
\frac{2}{\pi}-\frac{5\pi}{12n^2}+O(\frac{1}{n^4})~~&\text{if $n$ is odd}\\
\frac{2}{\pi}-\frac{2\pi}{3n^2}+O(\frac{1}{n^4})) &\text{if $n$ is even}.
\end{cases}
\end{align} 
Since $\frac{1}{n}\cot(\frac{\pi}{2n})\cos(\frac{\pi}{2n})$ is much tighter comparing to $\frac{2}{n}\cot(\frac{\pi}{n})$ in the asymptotical region, we will always lower bound $R^p(n)$ using the odd $n$ strategy (with a ``$n-1$'' shift to obtain lower bounds for both the even and odd cases.):
\begin{align}
R^p(n)\ge R^p(\{\Pi^{\rot}_{\lambda}\})>\frac{2}{\pi}-\frac{5\pi}{12(n-1)^2}, \label{eq: planar-poly-R}
\end{align}

Combing Eq.~\eqref{eq: planar-poly-R} with Eq. \eqref{Eq:compatibility-radius-complexity} and Proposition~\ref{prop:planar CC}, we can hence obtain the following upper bound of $N^p(r)$ and $\gamma^p(\omega_r)$: 
\begin{corollary}
For any $r\in[0,\frac{2}{\pi}]$, 
\begin{equation}
\gamma^p(\omega_r)=N^p(r)\leq \sqrt{\frac{5\pi}{12}}\left(\frac{2}{\pi}-r\right)^{-\frac{1}{2}}+1,
\end{equation}
\label{cor:upper planer}
\end{corollary}
\blk
\par 
A noteworthy and somewhat counterintuitive features of $R^p(\{\Pi^{\rot}_{\lambda}\})$ is its non-monotonic in $n$, which implies that symmetric POVMs do not always have the largest compatibility radius (see also Table~\ref{tab: symmetric POVM}), contrary to common intuition \cite{Nguyen2019, Werner1989}, where symmetric strategy were typically used to construct hidden variable models.

Nevertheless, a comparison of Corollaries \ref{cor:lower planer} and \ref{cor:upper planer} demonstrates that $\{\Pi^{\rot}_{\lambda}\}_{{\lambda}=1}^n$ is \yujie asymptotically \blk optimal for simulating noisy spin measurements in the $x$-$z$ plane.
\begin{table}[t]
    \centering
    \begin{tabular}{ccccc}
    \hline
    \hline
         &  \multicolumn{2}{c}{ Planar $R^p(n)$}  & \multicolumn{2}{c}{ General $R(n)$}  \\
         \hline
     n    & Symmetric &  Numerics & Thomson  & Numerics\\
     \hline
     3     &   0.5   &  0.5 &  0  & 0 \\
     4    &   0.5   &   0.5274  &  0.3333&  0.3333 \\
     5    &   0.5854    & 0.5854&  0.3464 &  0.3718 \\
     6    &   0.5774    &  0.5927 & 0.3333 & 0.4004\\
     7    &   0.6102    &  0.6102  & 0.2857& 0.4060\\
     8    &   0.6035    & 0.6111  & 0.4392&    \\
     9    &   0.6206     & 0.6206  & 0.4446&\\  
     10   &   0.6155     &  0.6213 & 0.4376&\\
     11   &   0.6259     & 0.6259       &--&\\
     12   &   0.6220    &  0.6265    & 0.4588&\\
     \hline
     \hline
     \end{tabular}
    \caption{Lower bounds on the compatible radii $R^p(n)$ and $R(n)$.  For the planar case, the first column describes the largest radius achievable by symmetric POVMs, as given by Eq. \eqref{eq: planar sp}.  For the general case, the third column provides a lower bound on $R(n)$ using solutions to the Thomson problem, while tighter numerical bounds are obtained by random sampling the parent POVM $\{ \Pi_{\lambda}\} $\cite{Yujie2023}.}
        \label{tab: symmetric POVM}
\end{table}

\subsection{General qubit measurements}
\par
Next, we turn to simulating all noisy spin measurements. In this general scenario, the rotationally symmetric planar POVM $\{\Pi^{\text{rot}}_\lambda\}$ cannot be directly generalized, as noisy spin measurements are no longer represented by a shrunken 2D disk but rather a 3D shrunken ball in the Bloch sphere representation. To address this, we consider alternative symmetric constructions for general qubit measurements, based on the idea of "evenly distributing" points on a sphere.

First, we utilize the equally-spaced points obtained from the Thompson problem \cite{Wales2006}. The numerical results for $n\in\{4,\cdots, 12\}$ are shown in Table~\ref{tab: symmetric POVM}. Second, we construct a parent POVM by leveraging the symmetries of Platonic solids. A detailed analysis of this approach is provided in the appendix~\ref{appendixE}, with results summarized in Table~\ref{tab:platonic}.
 \begin{table}[h]
    \centering
    \begin{tabular}{c|c|c}
    \hline
    \hline
$n$ & POVM $\{ \Pi_{\lambda}\}$ &  $R(\{ \Pi_{\lambda}\})$\\
\hline
4 &  tetrahedron & $\frac{1}{3}$\\
\hline
6 &   Octahedron &  $\frac{1}{3}$\\
\hline
8 &   Cube & $\frac{\sqrt{6}}{6}$\\
\hline
12 &  Icosahedron & $\frac{\phi^3\sqrt{1+(1-\phi)^2}}{3({1+\phi^2})}$\\ 
\hline
 20 &  Dodecahedron &  $\sqrt{\frac{5}{6}}\frac{\phi^2}{5}$\\
 \hline
 \hline
    \end{tabular}
    \caption{Parent POVM with platonic configuration and their \yujie associated compatible radius $R(\{ \Pi_{\lambda}\})$ \blk, where $\phi=\frac{1+\sqrt{5}}{2}$ is the golden ratio.}
    \label{tab:platonic}
\end{table}
 
The final and most powerful construction leverages results from Refs. \cite{Bourgain1988, siegel2023} on zonotope approximations. These works discuss how to achieve the best approximation of a unit ball using a zonotope generated by a finite number of generators. This is particularly relevant because the compatible region $\msf{m}_{\{\Pi_{\lambda}\}}^*$ mentioned earlier in the preliminaries is, in fact, a zonotope. \par 

A detailed discussion on zonotopes is provided in the appendix~\ref{appendixC}, and the key result is summarized in the following.
\yujie 
\begin{theorem}
\label{thm: zonotope approx}
For any $r\in[0,\frac{1}{2}]$,
\begin{align}
\label{Eq:upb-full}
R(2n)&\ge \frac{1}{2}-Cn^{-\frac{5}{4}}.\\
N(r)&\leq C'\left(\frac{1}{2}-r\right)^{-\frac{4}{5}}
\end{align}
for some positive constants $C$ and $C'=2C^{\frac{4}{5}}$ originated from the study of zonotope approximation problem \cite{Bourgain1988, siegel2023}, and will be briefly discussed in Appendix~\ref{appendixC}.
\end{theorem}
\blk 
However, this does not directly provide an upper bound for the simulation cost of $\gamma(\omega_r)$, as Proposition~\ref{prop:general CC} is not an equality.
\ref{thm: zonotope approx}

\section{Simulating Cost of General Unsteerable States} 
Corollary \ref{Cor: Main} shows that the separability threshold for two-qubit Werner states coincides with a jump in shared randomness cost from $4$ to $5$ (i.e., 2 bits to $\log_25$ bits).  In fact, this connection between separability and simulation cost holds for all two-qubit states $\rho_{AB}$. That is, $\gamma(\rho_{AB})>4$ if and only if $\rho_{AB}$ is entangled.  As explained below, this previously unrecognized result actually follows directly from two different results in quantum information theory. First, it happens that for two-qubit separable states, the number of product states needed in a separable decomposition is always at most $4$ \cite{Sanpera1998, Wootters1998}. Second, Jevtic \textit{et al.} have provided a `nested tetrahedron' condition \cite{Sania2014}, which, when combined with the `four-packable' condition by Nguyen and Vu \cite{Nguyen2016a}, says the following.

\yujie
\begin{lemma}[\cite{Sania2014}] 
A two-qubit state $\rho_{AB}$ is separable if and only if its steering ellipsoid $\mc{E}_B$ fits inside a tetrahedron that fits inside the Bloch sphere
\label{lem: nest tetrahedron}
\end{lemma}
\noindent To be more precise, the steering ellipsoid $\mc{E}_B$ for a bipartite state $\rho_{AB}$ is defined by first transforming $\rho_{AB}$ to $\tilde{\rho}_{AB}\propto [(\rho_A)^{-1/2}\otimes \mbb{I}]\rho_{AB}[(\rho_A)^{-1/2}\otimes \mbb{I}]$ with $\rho_{A}=\tr_B[\rho_{AB}]$ (assumes $\rho_{A}$ has full-rank, otherwise, it is just a product state). The possible unnormalized states that it can be steered to using the measurement effects $M_{\pm|\hat{n}}=\frac{1}{2}(\mbb{I}\pm \hat{n}\cdot \vec{\sigma} )$ as in Eq.~\ref{Eq:channel-conversion} are given as 

\begin{align}
&\tilde{\rho}_{AB}=\begin{pmatrix}
   1 &  \vec{\tilde{b}} \\
   \vec{0} & \tilde{T}  
   \end{pmatrix}
   \xRightarrow[]{M_{\pm|\hat{n}}}\sigma_{\pm|\hat{n}}=\frac{1}{4}\begin{pmatrix}
      1\\
      \vec{\tilde{b}}\pm \tilde{T}\hat{n}
   \end{pmatrix} \notag, \\
\end{align}
where we have written down both $\tilde{\rho}_{AB}$ the unnormalized state  $\sigma_{\pm|\hat{n}}$ as four-dimensional vectors in Pauli basis $\{\mbb{I},\sigma_x,\sigma_y,\sigma_z\}$, and the steering ellipsoid is then defined as the set
\begin{equation}
     \mc{E}_B=\{ \vec{\tilde{b}}+\tilde{T}\hat{n}\}_{\hat{n}}
\end{equation}
\begin{proposition}
\yujie A two-qubit state $\rho_{AB}$ is entangled if and only if $\gamma(\rho_{AB})>4$. 
\end{proposition}
\begin{proof}
The ``only if'' part follows directly from the existence of minimal separable decompositions for two-qubit state \cite{Sanpera1998, Wootters1998}, which is no more than 4.  \par  

The ``if'' part can also be simply proved by contradiction.  Assume $\rho_{AB}$ is entangled (thus must have full-rank marginal) but has a simulation cost less than or equal to $4$, then the filtered state $\tilde{\rho}_{AB}\propto [(\rho_A)^{-1/2}\otimes \mbb{I}]\rho_{AB}[(\rho_A)^{-1/2}\otimes \mbb{I}]$ must also be entangled, and with simulation cost less than or equal to $4$. This is because the full-rank local filter on Alice is invertible and completely positive; it preserves separability/entanglement and induces a bijection between LHS models.

Therefore, there exists a set of local hidden states $\{\rho_{\lambda}=\frac{1}{2}(\mbb{I}+\vec{n}_{\lambda}\cdot\vec{\sigma})\}_{\lambda=1}^4$ such that for measurement effects $M_{\pm|\hat{n}}=\frac{1}{2}(\mbb{I}\pm \hat{n}\cdot \vec{\sigma} )$ acting on $\tilde{\rho}_{AB}$, the unnormalized state $\sigma_{\pm|\hat{n}}$
can be simulated by
\begin{align}
&\sigma_{\pm|\hat{n}}=\sum_{\lambda=1}^4p(\pm|\hat{n},\lambda)p(\lambda)\rho_{\lambda} \notag \\
\Rightarrow &\mc{E}_B=\{ \vec{\tilde{b}}+\tilde{T}\hat{n}\}_{\hat{n}}\subseteq \text{Conv}[\vec{n}_{\lambda}]_{\lambda=1}^4.
\end{align}
This essentially shows that there exists a tetrahedron defined by the Bloch vectors $\{\vec{n}_{\lambda}\}_{\lambda=1}^4$ of  $\{\rho_{\lambda}\}_{\lambda=1}^4$ that contain the ellipsoid $\mc{E}_B$. Thus, $\tilde{\rho}_{AB}$ is separable from Lemma ~\ref{lem: nest tetrahedron}, and there is a contradiction. \par 
Thus, we conclude that a two-qubit $\rho_{AB}$ is entangled if and only if $\gamma(\rho_{AB})>4$.
\end{proof}
\begin{remark}
For two qubits the same threshold holds symmetrically if one considers steering Alice by measuring Bob; choosing which side is measured does not affect the conclusion.
\end{remark}
\blk

\yujie {We now discuss ways in which our results can be applied to unsteerable states beyond two-qubit Werner family.  Observe that any unsteerable state $\rho_{AB}$ will satisfy $\gamma(\rho)\geq \gamma(\mc{L}(\rho))/m$, i.e., $\log_2(\gamma(\rho))\geq \log_2(\gamma(\mc{L}(\rho)))-\log_2m $ under any map of the form $\mc{L}=\sum_{i=1}^m\mc{N}_i\otimes\mc{E}_i$ ({one-way LOCC map from Bob to Alice}); where the $\mc{N}_i$ are completely positive and trace-preserving (CPTP) maps while the $\mc{E}_i$ are completely positive {and $\sum_{i}\mc{E}_i$ being trace-preserving.  Indeed, suppose that $\tr_A[M_{a|x}\otimes\mbb{I}\rho_{AB}]=\sum_{\lambda=1}^tp(a|x,\{M_{a|x}\},\lambda)p(\lambda)\rho_\lambda$ is a LHS model for $\rho_{AB}$, with $\mc{M}=\{M_{a|x}\}$ being an arbitrary family of measurements.  Then for any CP map $\mc{E}$ on Bob's system, we also have $\tr_A[(M_{a|x}\otimes\mbb{I})\mc{E}_B(\rho_{AB})]=\sum_{\lambda=1}^tp(a|x,\{M_{a|x}\},\lambda)p(\lambda)\mc{E}(\rho_\lambda)$, and so
\begin{align}
    \tr_A&[(M_{a|x}\otimes\mbb{I})\mc{L}(\rho_{AB})]\notag\\
    &=\sum_{i=1}^m\tr_A[(M_{a|x}\otimes\mbb{I})\mc{N}_i\otimes \mc{E}_i(\rho_{AB})]\notag\\
    &=\sum_{i=1}^m\tr_A[(\mc{N}^\dagger_i(M_{a|x})\otimes\mbb{I})\text{id}\otimes \mc{E}_i(\rho_{AB})]\notag\\
    &=\sum_{i=1}^m\sum_{\lambda=1}^t p(a|x,\{\mc{N}_i^\dagger(M_{a|x})\},\lambda)p(\lambda)\mc{E}_i(\rho_\lambda)\notag\\
     &=\sum_{i=1}^m\sum_{\lambda=1}^t q(a|x,\{M_{a|x}\},\lambda,i)q(\lambda,i)\sigma_{\lambda,i}
     \label{Eq:LHS-one-wayLOCC}
\end{align}
where $q(a|x,\{M_{a|x}\},\lambda,i)=p(a|x,\{\mc{N}_i^\dagger(M_{a|x})\},\lambda)$,  $q(\lambda,i)=p(\lambda)\tr[\mc{E}_i(\rho_\lambda)]$, and $\sigma_{\lambda_i}=\mc{E}_i(\rho_\lambda)/\tr[\mc{E}_i(\rho_\lambda)]$. Hence, we can lower bound the simulation cost of $\rho_{AB}$ by first transforming it to a two-qubit state using the map above}. 

Secondly, it is well-known that every bipartite state can be converted into a Werner state when Alice and Bob simultaneously apply one of twelve random Clifford gates to their system (a so-called ``twirling map'') \cite{DiVincenzo-2002a}.  Moreover, the noise parameter $r$ in the resulting state is given by the singlet fraction, i.e. $\rho\mapsto \rho_W(r)$ with $r=\frac{4\bra{\Phi^-}\rho\ket{\Phi^-}-1}{3}$.  Since the twirling map will increase the shared randomness in an LHS model by a factor of twelve, we conclude that any unsteerable state $\rho_{AB}$ with $\bra{\Phi^-}\mc{L}(\rho)\ket{\Phi^-}>\frac{3{R}(n)+1}{4}$ will require an LHS model with at least $\log_2 n-\log_2(12m)$ bits of shared randomness{, by the same reasoning used in Eq. \eqref{Eq:LHS-one-wayLOCC}.} \yujie 
Otherwise, one can use $\log_2(12m)$ bit share randomness to turn the state into a two-qubit Werner state and find a better way to simulate two-qubit Werner state.  From Theorem \ref{thm:SC lowerbound}, this amount grows unbounded as $\bra{\Phi^-}\mc{L}(\rho)\ket{\Phi^-}\to \frac{5}{8}$, which corresponds to the threshold $r=\frac{1}{2}$ for Werner states.  \blk

\textit{Conclusions -- }  In this work we studied the shared randomness cost $\gamma(\rho_{AB})$ for simulating c-to-cq channels built using $\rho_{AB}$. We characterized the simulation cost and showed that it can be unbounded. The mathematical methods we developed to study this problem used the correspondence between steerability and measurement incompatibility and their similar geometric picture.  Analogous to any other type of channel simulation cost, the quantity $\gamma(\rho_{AB})$ provides one measure of operational resourcefulness for every unsteerable state $\rho$.  Furthermore, understanding the simulation cost of different assemblages can be used for semi-device-independent entanglement verification \cite{Cavalcanti-2017a} when Alice and Bob are known to have a limited amount of shared randomness, even in noisy environments where quantum steering is not possible.

The results presented here pertain to LHS models.  As future work, it would be interesting if similar techniques can be applied to bound the shared randomness needed for local hidden variable (LHV) models.  Since LHS models are more demanding than LHV models, the upper bounds presented in this work will still apply.

Finally, we close by conjecturing a more fundamental role for the simulation cost $\gamma(\rho_{AB})$ in the study of nonlocality and entanglement.  The notion of steerability provides a clean partitioning in the set of bipartite quantum states between those that are steerable versus those that are not.  However, this partition does not coincide with the separation between entangled and separable states \cite{Barrett-2002a, Wiseman2007}.  In contrast, we have found here that the for two-qubit states the jump from separable to entangled-unsteerable coincides with a jump in the simulation cost.  Hence, by moving beyond steering and focusing on simulation cost, we recover the separable/entangled boundary.  Perhaps such a relationship holds for higher-dimensional bipartite states as well.

\textit{Acknowledgements—} This work was supported by NSF Award No. 1839177. We thank Virginia Lorenz and Marius Junge for helpful discussions during the preparation of this manuscript. \yujie We also thank the anonymous reviewer for carefully reading the manuscript and for pointing out errors and typos in earlier versions. \blk

\bibliographystyle{quantum}
\bibliography{complexity}

\begin{thebibliography}{10}

\bibitem{Cubitt-2010a}
Toby~S. Cubitt, Debbie Leung, William Matthews, and Andreas Winter.
\newblock ``Improving zero-error classical communication with entanglement''.
\newblock \href{https://dx.doi.org/10.1103/PhysRevLett.104.230503}{Phys. Rev. Lett. {\bf 104}, 230503}~(2010).

\bibitem{Leung-2012a}
Debbie Leung, Laura Mancinska, William Matthews, Maris Ozols, and Aidan Roy.
\newblock ``Entanglement can increase asymptotic rates of zero-error classical communication over classical channels''.
\newblock \href{https://dx.doi.org/10.1007/s00220-012-1451-x}{Communications in Mathematical Physics {\bf 311}, 97--111}~(2012).

\bibitem{Bennett-1993a}
Charles~H. Bennett, Gilles Brassard, Claude Cr\'epeau, Richard Jozsa, Asher Peres, and William~K. Wootters.
\newblock ``Teleporting an unknown quantum state via dual classical and einstein-podolsky-rosen channels''.
\newblock \href{https://dx.doi.org/10.1103/PhysRevLett.70.1895}{Phys. Rev. Lett. {\bf 70}, 1895--1899}~(1993).

\bibitem{Bennett-2002a}
C.H. Bennett, P.W. Shor, J.A. Smolin, and A.V. Thapliyal.
\newblock ``Entanglement-assisted capacity of a quantum channel and the reverse shannon theorem''.
\newblock \href{https://dx.doi.org/10.1109/tit.2002.802612}{{IEEE} Transactions on Information Theory {\bf 48}, 2637--2655}~(2002).

\bibitem{Cubitt-2011a}
Toby~S. Cubitt, Debbie Leung, William Matthews, and Andreas Winter.
\newblock ``Zero-error channel capacity and simulation assisted by non-local correlations''.
\newblock \href{https://dx.doi.org/10.1109/tit.2011.2159047}{{IEEE} Transactions on Information Theory {\bf 57}, 5509--5523}~(2011).

\bibitem{Schmid-2020a}
David Schmid, Denis Rosset, and Francesco Buscemi.
\newblock ``The type-independent resource theory of local operations and shared randomness''.
\newblock \href{https://dx.doi.org/10.22331/q-2020-04-30-262}{Quantum {\bf 4}, 262}~(2020).

\bibitem{Uola2020}
Roope Uola, Ana C.~S. Costa, H.~Chau Nguyen, and Otfried G\"uhne.
\newblock ``Quantum steering''.
\newblock \href{https://dx.doi.org/10.1103/RevModPhys.92.015001}{Rev. Mod. Phys. {\bf 92}, 015001}~(2020).

\bibitem{Wiseman2007}
H.~M. Wiseman, S.~J. Jones, and A.~C. Doherty.
\newblock ``Steering, entanglement, nonlocality, and the einstein-podolsky-rosen paradox''.
\newblock \href{https://dx.doi.org/10.1103/PhysRevLett.98.140402}{Phys. Rev. Lett. {\bf 98}, 140402}~(2007).

\bibitem{Cavalcanti-2017a}
D~Cavalcanti and P~Skrzypczyk.
\newblock ``Quantum steering: a review with focus on semidefinite programming''.
\newblock \href{https://dx.doi.org/10.1088/1361-6633/80/2/024001}{Reports on Progress in Physics {\bf 80}, 024001}~(2016).

\bibitem{Piani-2015a}
Marco Piani and John Watrous.
\newblock ``Necessary and sufficient quantum information characterization of einstein-podolsky-rosen steering''.
\newblock \href{https://dx.doi.org/10.1103/PhysRevLett.114.060404}{Phys. Rev. Lett. {\bf 114}, 060404}~(2015).

\bibitem{Bennett2014}
Charles~H. Bennett, Igor Devetak, Aram~W. Harrow, Peter~W. Shor, and Andreas Winter.
\newblock ``The quantum reverse shannon theorem and resource tradeoffs for simulating quantum channels''.
\newblock \href{https://dx.doi.org/10.1109/tit.2014.2309968}{IEEE Transactions on Information Theory {\bf 60}, 2926–2959}~(2014).

\bibitem{Toner2003}
B.~F. Toner and D.~Bacon.
\newblock ``Communication cost of simulating bell correlations''.
\newblock \href{https://dx.doi.org/10.1103/PhysRevLett.91.187904}{Phys. Rev. Lett. {\bf 91}, 187904}~(2003).

\bibitem{Renner2023}
Martin~J. Renner, Armin Tavakoli, and Marco~T\'ulio Quintino.
\newblock ``Classical cost of transmitting a qubit''.
\newblock \href{https://dx.doi.org/10.1103/PhysRevLett.130.120801}{Phys. Rev. Lett. {\bf 130}, 120801}~(2023).

\bibitem{Sainz2016}
Ana~Bel\'en Sainz, Leandro Aolita, Nicolas Brunner, Rodrigo Gallego, and Paul Skrzypczyk.
\newblock ``Classical communication cost of quantum steering''.
\newblock \href{https://dx.doi.org/10.1103/PhysRevA.94.012308}{Phys. Rev. A {\bf 94}, 012308}~(2016).

\bibitem{Cavalcanti}
Daniel Cavalcanti, Paul Skrzypczyk, and Ivan \ifmmode \check{S}\else \v{S}\fi{}upi\ifmmode~\acute{c}\else \'{c}\fi{}.
\newblock ``All entangled states can demonstrate nonclassical teleportation''.
\newblock \href{https://dx.doi.org/10.1103/PhysRevLett.119.110501}{Phys. Rev. Lett. {\bf 119}, 110501}~(2017).

\bibitem{Piani2009}
Marco Piani and John Watrous.
\newblock ``All entangled states are useful for channel discrimination''.
\newblock \href{https://dx.doi.org/10.1103/PhysRevLett.102.250501}{Phys. Rev. Lett. {\bf 102}, 250501}~(2009).

\bibitem{Hirsch2016}
Flavien Hirsch, Marco~T\'ulio Quintino, Tam\'as V\'ertesi, Matthew~F. Pusey, and Nicolas Brunner.
\newblock ``Algorithmic construction of local hidden variable models for entangled quantum states''.
\newblock \href{https://dx.doi.org/10.1103/PhysRevLett.117.190402}{Phys. Rev. Lett. {\bf 117}, 190402}~(2016).

\bibitem{Bowles2015}
Joseph Bowles, Flavien Hirsch, Marco~T\'ulio Quintino, and Nicolas Brunner.
\newblock ``Local hidden variable models for entangled quantum states using finite shared randomness''.
\newblock \href{https://dx.doi.org/10.1103/PhysRevLett.114.120401}{Phys. Rev. Lett. {\bf 114}, 120401}~(2015).

\bibitem{Werner1989}
Reinhard~F. Werner.
\newblock ``Quantum states with einstein-podolsky-rosen correlations admitting a hidden-variable model''.
\newblock \href{https://dx.doi.org/10.1103/PhysRevA.40.4277}{Phys. Rev. A {\bf 40}, 4277--4281}~(1989).

\bibitem{Buscemi-2020a}
Francesco Buscemi, Eric Chitambar, and Wenbin Zhou.
\newblock ``Complete resource theory of quantum incompatibility as quantum programmability''.
\newblock \href{https://dx.doi.org/10.1103/PhysRevLett.124.120401}{Phys. Rev. Lett. {\bf 124}, 120401}~(2020).

\bibitem{zhang2024towards}
Yujie Zhang.
\newblock ``Towards quantum networks: Theory, experiment, and applications''.
\newblock Ph.d. dissertation.
\newblock University of Illinois at Urbana–Champaign.
\newblock Urbana, IL~(2024).
\newblock  url:~\url{https://hdl.handle.net/2142/124212}.

\bibitem{Bowles2014}
Joseph Bowles, Tam\'as V\'ertesi, Marco~T\'ulio Quintino, and Nicolas Brunner.
\newblock ``One-way einstein-podolsky-rosen steering''.
\newblock \href{https://dx.doi.org/10.1103/PhysRevLett.112.200402}{Phys. Rev. Lett. {\bf 112}, 200402}~(2014).

\bibitem{Bowles2016}
Joseph Bowles, Flavien Hirsch, Marco~T\'ulio Quintino, and Nicolas Brunner.
\newblock ``Sufficient criterion for guaranteeing that a two-qubit state is unsteerable''.
\newblock \href{https://dx.doi.org/10.1103/PhysRevA.93.022121}{Phys. Rev. A {\bf 93}, 022121}~(2016).

\bibitem{Bavaresco2017}
Jessica Bavaresco, Marco~T\'ulio Quintino, Leonardo Guerini, Thiago~O. Maciel, Daniel Cavalcanti, and Marcelo~Terra Cunha.
\newblock ``Most incompatible measurements for robust steering tests''.
\newblock \href{https://dx.doi.org/10.1103/PhysRevA.96.022110}{Phys. Rev. A {\bf 96}, 022110}~(2017).

\bibitem{Uhlmann-1998a}
Armin Uhlmann.
\newblock ``Entropy and optimal decompositions of states relative to a maximal commutative subalgebra''.
\newblock \href{https://dx.doi.org/10.1023/A:1009664331611}{Open Systems {\&} Information Dynamics {\bf 5}, 209--228}~(1998).

\bibitem{Quintino2014}
Marco~T\'ulio Quintino, Tam\'as V\'ertesi, and Nicolas Brunner.
\newblock ``Joint measurability, einstein-podolsky-rosen steering, and bell nonlocality''.
\newblock \href{https://dx.doi.org/10.1103/PhysRevLett.113.160402}{Phys. Rev. Lett. {\bf 113}, 160402}~(2014).

\bibitem{Uola2014}
Roope Uola, Tobias Moroder, and Otfried G\"uhne.
\newblock ``Joint measurability of generalized measurements implies classicality''.
\newblock \href{https://dx.doi.org/10.1103/PhysRevLett.113.160403}{Phys. Rev. Lett. {\bf 113}, 160403}~(2014).

\bibitem{Brunner2014}
Nicolas Brunner, Daniel Cavalcanti, Stefano Pironio, Valerio Scarani, and Stephanie Wehner.
\newblock ``Bell nonlocality''.
\newblock \href{https://dx.doi.org/10.1103/RevModPhys.86.419}{Rev. Mod. Phys. {\bf 86}, 419--478}~(2014).

\bibitem{Heinosaari-2016a}
Teiko Heinosaari, Takayuki Miyadera, and M{\'{a}}rio Ziman.
\newblock ``An invitation to quantum incompatibility''.
\newblock \href{https://dx.doi.org/10.1088/1751-8113/49/12/123001}{Journal of Physics A: Mathematical and Theoretical {\bf 49}, 123001}~(2016).

\bibitem{Guerini-2017a}
Leonardo Guerini, Jessica Bavaresco, Marcelo~Terra Cunha, and Antonio Ac{\'{\i}}n.
\newblock ``Operational framework for quantum measurement simulability''.
\newblock \href{https://dx.doi.org/10.1063/1.4994303}{Journal of Mathematical Physics {\bf 58}, 092102}~(2017).

\bibitem{Skrzypczyk2020}
Paul Skrzypczyk, Matty~J. Hoban, Ana~Bel\'en Sainz, and Noah Linden.
\newblock ``Complexity of compatible measurements''.
\newblock \href{https://dx.doi.org/10.1103/PhysRevResearch.2.023292}{Phys. Rev. Research {\bf 2}, 023292}~(2020).

\bibitem{Heinosaari-2015a}
Teiko Heinosaari, Jukka Kiukas, Daniel Reitzner, and Jussi Schultz.
\newblock ``Incompatibility breaking quantum channels''.
\newblock \href{https://dx.doi.org/10.1088/1751-8113/48/43/435301}{Journal of Physics A: Mathematical and Theoretical {\bf 48}, 435301}~(2015).

\bibitem{Uola2016}
Roope Uola, Kimmo Luoma, Tobias Moroder, and Teiko Heinosaari.
\newblock ``Adaptive strategy for joint measurements''.
\newblock \href{https://dx.doi.org/10.1103/PhysRevA.94.022109}{Phys. Rev. A {\bf 94}, 022109}~(2016).

\bibitem{Zhang2024}
Yujie Zhang and Eric Chitambar.
\newblock ``Exact steering bound for two-qubit werner states''.
\newblock \href{https://dx.doi.org/10.1103/PhysRevLett.132.250201}{Phys. Rev. Lett. {\bf 132}, 250201}~(2024).

\bibitem{Renner2024}
Martin~J. Renner.
\newblock ``Compatibility of generalized noisy qubit measurements''.
\newblock \href{https://dx.doi.org/10.1103/PhysRevLett.132.250202}{Phys. Rev. Lett. {\bf 132}, 250202}~(2024).

\bibitem{Cavalcanti2016}
D~Cavalcanti and P~Skrzypczyk.
\newblock ``Quantum steering: a review with focus on semidefinite programming''.
\newblock \href{https://dx.doi.org/10.1088/1361-6633/80/2/024001}{Reports on Progress in Physics {\bf 80}, 024001}~(2016).

\bibitem{McMullen1971}
Peter McMullen.
\newblock ``On zonotopes''.
\newblock \href{https://dx.doi.org/10.1090/S0002-9947-1971-0279689-2}{Transactions of the American Mathematical Society {\bf 159}, 91--109}~(1971).

\bibitem{ziegler2012}
G{\"u}nter~M. Ziegler.
\newblock ``Lectures on polytopes''.
\newblock \href{https://dx.doi.org/10.1007/978-1-4613-8431-1}{Volume 152 of Graduate Texts in Mathematics}.
\newblock Springer. New York, NY~(2012).

\bibitem{Murray1979}
Murray~S. Klamkin and George~A. Tsintsifas.
\newblock ``The circumradius-inradius inequality for a simplex''.
\newblock Mathematics Magazine {\bf 52}, 20--22~(1979).
\newblock  url:~\url{http://www.jstor.org/stable/2689968}.

\bibitem{blaasjo2005}
Viktor Bl{\aa}sj{\"o}.
\newblock ``The isoperimetric problem''.
\newblock \href{https://dx.doi.org/10.1080/00029890.2005.11920227}{The American Mathematical Monthly {\bf 112}, 526--566}~(2005).

\bibitem{Bourgain1988}
J.~Bourgain and J.~Lindenstrauss.
\newblock ``Distribution of points on spheres and approximation by zonotopes''.
\newblock \href{https://dx.doi.org/10.1007/BF02767366}{Israel Journal of Mathematics {\bf 64}, 25--31}~(1988).

\bibitem{Bourgain1993}
Jean Bourgain and Joram Lindenstrauss.
\newblock ``Approximating the ball by a minkowski sum of segments with equal length''.
\newblock \href{https://dx.doi.org/10.1007/BF02189313}{Discrete {\&} Computational Geometry {\bf 9}, 131--144}~(1993).

\bibitem{Nguyen2019}
H.~Chau Nguyen, Huy-Viet Nguyen, and Otfried G\"uhne.
\newblock ``Geometry of einstein-podolsky-rosen correlations''.
\newblock \href{https://dx.doi.org/10.1103/PhysRevLett.122.240401}{Phys. Rev. Lett. {\bf 122}, 240401}~(2019).

\bibitem{Yujie2023}
Yujie Zhang.
\newblock ``{Compatible-radius-calculation}''.
\newblock \url{https://github.com/yujie4phy/Compatible-radius-calculation}~(2023).
\newblock GitHub repository.

\bibitem{Wales2006}
David~J. Wales and Sidika Ulker.
\newblock ``Structure and dynamics of spherical crystals characterized for the thomson problem''.
\newblock \href{https://dx.doi.org/10.1103/PhysRevB.74.212101}{Phys. Rev. B {\bf 74}, 212101}~(2006).

\bibitem{siegel2023}
Jonathan~W. Siegel.
\newblock ``Optimal approximation of zonoids and uniform approximation by shallow neural networks''.
\newblock \href{https://dx.doi.org/10.1007/s00365-025-09712-9}{Constructive Approximation {\bf 62}, 441--469}~(2025).

\bibitem{Sanpera1998}
Anna Sanpera, Rolf Tarrach, and Guifr\'e Vidal.
\newblock ``Local description of quantum inseparability''.
\newblock \href{https://dx.doi.org/10.1103/PhysRevA.58.826}{Phys. Rev. A {\bf 58}, 826--830}~(1998).

\bibitem{Wootters1998}
William~K. Wootters.
\newblock ``Entanglement of formation of an arbitrary state of two qubits''.
\newblock \href{https://dx.doi.org/10.1103/PhysRevLett.80.2245}{Phys. Rev. Lett. {\bf 80}, 2245--2248}~(1998).

\bibitem{Sania2014}
Sania Jevtic, Matthew Pusey, David Jennings, and Terry Rudolph.
\newblock ``Quantum steering ellipsoids''.
\newblock \href{https://dx.doi.org/10.1103/PhysRevLett.113.020402}{Phys. Rev. Lett. {\bf 113}, 020402}~(2014).

\bibitem{Nguyen2016a}
H.~Chau Nguyen and Thanh Vu.
\newblock ``Nonseparability and steerability of two-qubit states from the geometry of steering outcomes''.
\newblock \href{https://dx.doi.org/10.1103/PhysRevA.94.012114}{Phys. Rev. A {\bf 94}, 012114}~(2016).

\bibitem{DiVincenzo-2002a}
D.P. DiVincenzo, D.W. Leung, and B.M. Terhal.
\newblock ``Quantum data hiding''.
\newblock \href{https://dx.doi.org/10.1109/18.985948}{{IEEE} Transactions on Information Theory {\bf 48}, 580--598}~(2002).

\bibitem{Barrett-2002a}
Jonathan Barrett.
\newblock ``Nonsequential positive-operator-valued measurements on entangled mixed states do not always violate a bell inequality''.
\newblock \href{https://dx.doi.org/10.1103/PhysRevA.65.042302}{Phys. Rev. A {\bf 65}, 042302}~(2002).

\bibitem{scott2016}
Joseph~K. Scott, Davide~Martino Raimondo, Giuseppe~Roberto Marseglia, and Richard~D. Braatz.
\newblock ``Constrained zonotopes: A new tool for set-based estimation and fault detection''.
\newblock \href{https://dx.doi.org/10.1016/j.automatica.2016.02.036}{Automatica {\bf 69}, 126--136}~(2016).

\bibitem{Bourgain1989}
J.~Bourgain, J.~Lindenstrauss, and V.~Milman.
\newblock ``{Approximation of zonoids by zonotopes}''.
\newblock \href{https://dx.doi.org/10.1007/BF02392835}{Acta Mathematica {\bf 162}, 73 -- 141}~(1989).

\end{thebibliography}
\onecolumngrid

\newpage
\begin{appendices}
\section*{Appendix}
\renewcommand{\thesubsection}{\arabic{subsection}}
\addcontentsline{toc}{section}{Supplementary Material}
  \startcontents
  \printcontents{}{1}{}
\vspace{0.5cm}
Here we provide some technical details that complements the main manuscript.\\
In section~\ref{appendixA} we establish a detailed geometrical connection between our problem and the analysis of zonotopes, a special type of convex polytope.
\\
In section~\ref{appendixB}, we prove prop.~\ref{Prop:Main} and cor.~\ref{Cor: Main} of the main text related to the first geometric inequality - the circumradius–inradius inequality of simplex. \\
In section~\ref{appendixD}, we connect our problem to the zonotope approximation problems and thm~\ref{thm:SC lowerbound} and thm~\ref{thm: zonotope approx} are proven with connection to geometric inequality therein, which allows us to study simulation cost in the asymptotic region (with cost $n\rightarrow \infty$).
\\
In section~\ref{appendixC}, we provide an optimization-based criterion for the computation of the compatibility radius $R(\{ \Pi_{\lambda}\})$ for a given $\{ \Pi_{\lambda}\}$, which can be used as lower bound for compatible radius $R(n)$, and thus for obtaining upper bound for $N(r)$. 
\\
In section~\ref{appendixE}, we give special examples of symmetric POVMs and calculate their compatibility region $\mathfrak{m}_{\{ \Pi_{\lambda}\}}$ and compatibility radius $R({\{ \Pi_{\lambda}\}})$ explicitly.   
\\In section~\ref{appendixF}, we establish the connection shared randomness cost in unsteerable state and compatible measurements. 
\\
Finally, in section~\ref{appendixG}, we provide evidence and conjecture that the simulation cost of noisy spin measurements (PVMs) might be strictly smaller than the simulation cost of noisy POVMs. 
\section{Geometry of compatible region}. 
\label{appendixA}
\begin{definition}
A zonotope is a set of points in $d$-dimensional space constructed from vectors $\{\vec{v}_{\lambda}\}$ by taking the Minkowski sum of line segments: $$\mathcal{Z}=\left
\{\sum_{\lambda} x_{\lambda}\vec{v}_{\lambda}|0\le x_{\lambda}\le 1\right\},$$ 
where the set of vectors $\{\vec{v}_i\}$ is called the generator of the zonotope.
\end{definition}
\par With the above notation, we introduce the zonotope defined by $\{\Pi_{\lambda}\}$ 
\begin{equation}
\mc{M}_{\{ \Pi_{\lambda}\}}:=\left\{\sum_{\lambda}q_{\lambda} \Pi_{\lambda}\bigg| 0\le q_{\lambda}\le 1\right\},
\end{equation}
which is the collection of all effects that can be simulated by $\{ \Pi_{\lambda}\}$ (or in other words, representing all dichotomic measurement that can be simulated by $\{ \Pi_{\lambda}\}$ ). This region can be viewed as a zonotope in an $d^2$-dimensional space of bounded operators \cite{McMullen1971, ziegler2012}, where $d$ is the dimension of the Hilbert space on which the measurement effects act. \par 

For qubit measurements, we can simplify the problem by parametrizing a quantum measurement $\{\Pi_{\lambda}\}$ in the Pauli basis, which provides a natural way to represent each effect $ \Pi_{\lambda}=\alpha_{\lambda}  (\mbb{I}+\eta_{\lambda}\hat{n}_{\lambda}\cdot\vec{\sigma})$ as a 4-dimensional vector:
\begin{equation}
\vec{\pi}_{\lambda}=(\alpha_{\lambda} ,\alpha_{\lambda}\eta_{\lambda} \hat{n}_{\lambda})^T=(\alpha_{\lambda} ,\alpha_{\lambda} \eta_{\lambda}\hat{n}_{\lambda}^x,\alpha_{\lambda} \eta_{\lambda}\hat{n}_{\lambda}^y,\alpha_{\lambda}\eta_{\lambda}\hat{n}_{\lambda}^z)^T.
\end{equation}
The normalization and positivity constraints become $\sum_{\lambda}\vec{\pi}_{\lambda}=(1,0,0,0)$ and $\alpha_{\lambda} \ge 0$. Thus, we can geometrically visualize the compatible region for any given measurement in a 4-dimensional Euclidean space:

\begin{equation}
    \mathfrak{M}_{\{ \Pi_{\lambda}\}}:=\left\{2\sum_{\lambda}q_{\lambda}\vec{\pi}_{\lambda}\bigg| 0\le q_{\lambda}\le 1\right\},
\end{equation}
where the factor of 2 is introduced for convenience.\par 

Compared to the 4-dimensional compatible region, there are two 3-dimensional subsets of it that play a crucial role in the study of the simulation cost of noisy spin measurements:
$\mathfrak{m}_{\{ \Pi_{\lambda}\}}$ and $\mathfrak{m}_{\{ \Pi_{\lambda}\}}^*$ defined in the main text. Both subsets can be derived from this 4-dimensional zonotope and are explained as follows:
\begin{definition}
\label{pro:constrained}
    The \textit{constrained zonotope} $\mathfrak{m}_{\{ \Pi_{\lambda}\}}:=\left\{2\sum_{\lambda}q_{\lambda}\alpha_{\lambda}\eta_{\lambda}\hat{n}_{\lambda}\bigg| 0\le q_{\lambda}\le 1, \sum_{\lambda} q_{\lambda}\alpha_{\lambda} = \frac{1}{2}\right\}$ can be the three-dimensional cross-section of $\mathfrak{M}_{\{ \Pi_{\lambda}\}}$ on the plane defined by $\sum_{\lambda} q_{\lambda}\alpha_{\lambda} = \frac{1}{2}$. However, in general, it need not be a zonotope\cite{scott2016}. The set  $\mathfrak{m}_{\{ \Pi_{\lambda}\}}$ represents the Bloch vectors of all unbiased two-outcome measurements (with each effect having the same trace) that can be simulated by $\{ \Pi_{\lambda}\}$.
\end{definition}
\begin{definition}
\label{pro:projected}
The \textit{projected zonotope}  $\mathfrak{m}^{*}_{\{ \Pi_{\lambda}\}}:=\left\{2\sum_{\lambda}q_{\lambda}\mbb{P}\vec{\pi}_{\lambda}\bigg| 0\le q_{\lambda}\le 1\right\}=\left\{2\sum_{\lambda}q_{\lambda}\alpha_{\lambda}\eta_{\lambda}\hat{n}_{\lambda}\bigg| 0\le q_{\lambda}\le 1\right\}$ is a zonotope, where $\mbb{P}$ projects vector $(x,\vec{y})^T\in\mbb{R}^4$ onto $(\vec{y})^T\in\mbb{R}^3$. By definition, $\mathfrak{m}_{\{ \Pi_{\lambda}\}}\subseteq \mathfrak{m}^{*}_{\{ \Pi_{\lambda}\}}$ with equality holding if and only if $\{\mbb{P}\vec{\pi}_{\lambda}\}=\{(\alpha_{\lambda} \eta_{\lambda}\hat{n}_{\lambda})^T\}$ is centrally symmetric.  Centrally symmetric measurements are the POVMs $\textbf{sym}\{ \Pi_{\lambda}\}$ introduced in the main text.
\end{definition}
The projected zonotope is meaningful both geometrically and analytically:
\begin{itemize}
    \item Geometrically, given any arbitrary POVM $\{ \Pi_{\lambda}\}$, we can symmetrically extend it to a new POVM $\textbf{sym}\{ \Pi_{\lambda}\}=\{\frac{ \Pi_{\lambda}}{2},\tr[\Pi_{\lambda}]\mbb{I}-\frac{ \Pi_{\lambda}}{2})\}$.  The projected zonotope for the original POVM is actually the constrained zonotope of the new symmetric-extended POVM $\textbf{sym}\{ \Pi_{\lambda}\}$.
    \item Analytically, when computing the compatibility radius $R(\{ \Pi_{\lambda}\})= \inf_{\substack{|\vec{c}|=1, -1\le c_0\le 1}}\sum_{\lambda}|\alpha_{\lambda} (c_0+\vec{c}\cdot\eta_{\lambda}\hat{n}_{\lambda})|$ for the constrained zonotope (see Appendix \ref{appendixC}), an upper bound is given by the compatible radius for the projected zonotope: $R(\{ \Pi_{\lambda}\})\le R_*(\{ \Pi_{\lambda}\}):= \inf_{\substack{|\vec{c}|=1}}\sum_{\lambda}|\alpha_{\lambda} (\vec{c}\cdot\eta_{\lambda}\hat{n}_{\lambda})|$, which appears to be easier to characterize in many cases.  
\end{itemize}
\yujie
To see  $\mathfrak{m}_{\textbf{sym}\{ \Pi_{\lambda}\}}=\mathfrak{m}^{*}_{\{ \Pi_{\lambda}\}}$, we note that,   $\mathfrak{m}_{\textbf{sym}\{ \Pi_{\lambda}\}\}}$ for POVM $\mathbf{sym}\{\Pi_\lambda\}
=\Bigl\{\tfrac{\alpha_\lambda}{2}\bigl(\mbb{I}\pm \eta_\lambda \hat n_\lambda\!\cdot\!\vec\sigma\bigr)\Bigr\}_\lambda$ is defined as:
\begin{align}
    \mathfrak{m}_{\textbf{sym}\{ \Pi_{\lambda}\}}:=&\left\{\sum_{\lambda,t\in\{\pm\}}tq_{\lambda,t}{\alpha_{\lambda}}\eta_{\lambda}\hat{n}_{\lambda}\bigg| 0\le q_{\lambda,t}\le 1, \sum_{\lambda,t\in\{\pm\}}q_{\lambda,t}\frac{\alpha_{\lambda}}{2} = \frac{1}{2}\right\} \notag \\
    =&\left\{\sum_{\lambda}(q_{\lambda,+}-q_{\lambda,-}){\alpha_{\lambda}}\eta_{\lambda}\hat{n}_{\lambda}\bigg| 0\le q_{\lambda,t}\le 1, \sum_{\lambda}\frac{(q_{\lambda,+}+q_{\lambda,-})}{2}\alpha_{\lambda} = \frac{1}{2}\right\}
\end{align}
\par 
In order to show that $\mathfrak{m}_{\textbf{sym}\{ \Pi_{\lambda}\}}\supseteq \mathfrak{m}^{*}_{\{ \Pi_{\lambda}\}}$. Given any $q_{\lambda}\in[0,1]$, we can choose $
q_{\lambda,\pm}:=\frac{1\pm (2q_{\lambda}-1)}{2}\in[0,1]$, which ensures that $\sum_{\lambda}\tfrac{\alpha_{\lambda}}{2}(q_{\lambda,+}+q_{\lambda,-})=\tfrac12$ and
\begin{equation}
\sum_{\lambda}(q_{\lambda,+}-q_{\lambda,-})\,\alpha_{\lambda}\eta_{\lambda}\hat n_{\lambda}
=\sum_{\lambda}(2q_{\lambda}-1)\,\alpha_{\lambda}\eta_{\lambda}\hat n_{\lambda}
=2\sum_{\lambda}q_{\lambda}\,\alpha_{\lambda}\eta_{\lambda}\hat n_{\lambda},
\end{equation} 
where in the last equality, we use the fact that $\sum_{\lambda}\alpha_{\lambda}\eta_{\lambda}\hat n_{\lambda}=\vec{0}$. Hence every point of $\mathfrak{m}^{*}_{\{\Pi_{\lambda}\}}$ lies in $\mathfrak{m}_{\mathbf{sym}\{\Pi_{\lambda}\}}$.
\par 
Similarly, to show that $\mathfrak{m}_{\textbf{sym}\{ \Pi_{\lambda}\}}\subseteq \mathfrak{m}^{*}_{\{ \Pi_{\lambda}\}}$.
Given any feasible $q_{\lambda,\pm}\in[0,1]$ with $\sum_{\lambda}\tfrac{\alpha_{\lambda}}{2}(q_{\lambda,+}+q_{\lambda,-})=\tfrac12$, we can define
\begin{equation}
q_{\lambda}:=\frac{1+(q_{\lambda,+}-q_{\lambda,-})}{2}\in[0,1].
\end{equation}
Then
\begin{equation}
\sum_{\lambda}(q_{\lambda,+}-q_{\lambda,-})\,\alpha_{\lambda}\eta_{\lambda}\hat n_{\lambda}
=\sum_{\lambda}(2q_{\lambda}-1)\,\alpha_{\lambda}\eta_{\lambda}\hat n_{\lambda}
=2\sum_{\lambda}q_{\lambda}\,\alpha_{\lambda}\eta_{\lambda}\hat n_{\lambda},
\end{equation}
again using $\sum_{\lambda}\alpha_{\lambda}\eta_{\lambda}\hat n_{\lambda}=\vec 0$. Thus every point of $\mathfrak{m}_{\mathbf{sym}\{\Pi_{\lambda}\}}$ lies in $\mathfrak{m}^{*}_{\{\Pi_{\lambda}\}}$.
\blk
\section{Proof of proposition~\ref{Prop:Main}, and Corollary~\ref{Cor: Main}}
\label{appendixB}
\yujie Here we sketch the proof of the circumradius–inradius inequality for an $n$-simplex, following Murray \cite{Murray1979}. As noted in the remark below, the result for $n=2,3$ can also be viewed as a direct consequence of the other identities/inequalities collected there.
\blk
\begin{lemma}[The circumradius-inradius inequality\cite{Murray1979}] The inradius $r$ of an arbitrary $n$-simplex is at least $n$ times less than it circumradius $R$.  The inequality saturates when the $n$-simplex is regular.
\label{lem:nsimplex}
\end{lemma}
\begin{proof}
Let $A_i$ and $F_i$ (for $i = 1,2, .. ., n + 1$) denote, respectively, the vertices and its opposite $(n - 1)$-dimensional faces of an $n$-dimensional simplex of volume $V$. Also, let $h_i$ and $\eta_{\lambda} $ denote the distances from $A_i$ and circumcenter $O$ to the $F_i$, respectively. \par
Then $R+\eta_{\lambda} \ge h_i$. Moreover, the volume of the $n$-dimensional simplex is given by $h_iF_i/n$, and so by evaluating the volume of the simplex in three different ways, we get:
$$nV=h_iF_i=\sum \eta_{\lambda} F_i=r\sum F_i.$$
Therefore,
$$\sum(R+\eta_{\lambda} )F_i=(R+r)\sum F_i\ge \sum h_iF_i=(n+1)r\sum F_i\rightarrow R\ge nr$$
\end{proof}
\yujie
\begin{remark}
For $n=2$, the Euler's theorem in geometry states that, the inradius 
$r$, the circumradius $R$ and the distance $d$ of the centers of the circumcicle and inscribed cicle are related by:
$$0\le d^2=R(R-2r)\Rightarrow R\ge 2r$$
For $n=3$, the Grace–Danielsson inequality states that, the inradius 
$r$, the circumradius $R$ and the distance $d$ of the centers of the circumcicle and inscribed cicle are related by:
$$0\le d^2\le (R+r)(R-3r)\Rightarrow R\ge 3r$$
\end{remark}
\blk

\setcounter{proposition}{2}
\begin{proposition}
\label{Prop:Main}
 $R^p(3)=\tfrac{1}{2}$ 
 and $R(4)=\tfrac{1}{3}$.
\end{proposition}
\begin{proof}
To apply Lemma~\ref{lem:nsimplex}, we begin with the planar case.  Given any three-outcome POVM $\{ \Pi_{\lambda}\}_{i=1}^3$, observe that the set $\mf{m}_{\{ \Pi_{\lambda}\}}$ is contained in the $2$-simplex $\mf{m}^{**}_{\{ \Pi_{\lambda}\}}=\{2\sum_{i=1}^3p_{\lambda}\hat{n}_{\lambda}\;|\;\sum_{\lambda=1}^3 p_{\lambda}=1, p_{\lambda}\ge 0\}$ in two dimensions. To see that, one can just replace $p_{\lambda}=2q_{\lambda}\alpha_{\lambda}$ in the constrained zonotope, and the constrained zonotope has more constrain on $0\le q_{\lambda}\ge 1$. The 2-simplex has circumradius one because each $\hat n_i$ is a unit vector. By the circumradius–inradius inequality,
\begin{align}
    R^p(3)\leq \sup_{\{ \Pi_{\lambda}\}_{i=1}^3}\!\inr(\mf{m}^{**}_{\{ \Pi_{\lambda}\}})=: R_{**}^p(3)\leq \frac{1}{2}.\notag
\end{align}
The same reasoning for four-outcome POVMs on the Bloch sphere gives $R(4)\leq \frac{1}{3}$, where now the compatible region $\mf{m}_{\{ \Pi_{\lambda}\}}$ and three-simplex $\mf{m}^{**}_{\{ \Pi_{\lambda}\}}$ are in $\mbb{R}^3$. 
Combined with the lower bounds from Table I in the main text, we have the stated equalities.   
\end{proof}
\setcounter{corollary}{0}
\begin{corollary}
For any $r>\tfrac{1}{3}$, we have $\gamma(\omega_r)> 4$. The simulation cost of any entangled Werner state is strictly greater than that of a separable Werner state.  
\label{Cor: Main}
\end{corollary}
\setcounter{proposition}{3}
\section{Proof of Theorem~\ref{thm:SC lowerbound}, Theorem~\ref{thm: zonotope approx} and connection to Zonotope approximation}
\label{appendixD}
In this section, we present the technical details showing how the well-studied results on approximating Euclidean balls with zonotopes apply to our problem.
\begin{lemma}
\label{prop:Bourgain1989}
  \yujie [proposition 6.6 \cite{Bourgain1989}] \blk For any zonotope generated by $2n$ line segments $\{\pm\alpha_{\lambda} \hat{n}_{\lambda}\}$ with $\sum_{\lambda}\alpha_{\lambda} =1$, There exists a positive constant $c_d$ depending only on dimension $d$ such that:
    \begin{equation}
    \norm{\sum_{\lambda}\alpha_{\lambda} |\langle\hat{n}_{\lambda},\hat{x}\rangle|-\beta_d}_{L^2(\sigma_d)}\ge c_dn^{-\frac{d+2}{2(d-1)}},
\end{equation}
where $\beta_3=\frac{1}{2}$, $\beta_2=\frac{2}{\pi}$ and $\norm{f}_{L^2(\sigma_{d})}=\sqrt{\int_{S_{d-1}}|f(\hat{x})|^2d\sigma_{d}(\hat{x})}$ is the $L^2$ norm on the integral of function over the $d-1$ dimensional unit sphere $S_{d-1}$. $\sigma_{d}(\hat{x})$ is the noramlized rotation invariant measure on $S_{d-1}$ with 
\begin{equation}
  \beta_d :=\int_{S_{d-1}}  |\langle\hat{n},\hat{x}\rangle| d\sigma_{d}(\hat{x})\quad\quad \forall~\hat{n}. \label{eq:defbeta}
\end{equation}
\end{lemma}
The proof is given in detail in \cite{Bourgain1989}  based on inequalities obtained from spherical harmonic expansion of these quantities. 
\setcounter{theorem}{0}
\begin{theorem}
For any $n$-outcome POVM $\{ \Pi_{\lambda}\}$, the compatible radius is upper bounded by
\begin{align}
R(n)\le \frac{1}{2}-\frac{1}{2}c_3^2n^{-\frac{5}{2}}&\Rightarrow \gamma(\omega_r)\ge c'_3|\frac{1}{2}-r|^{-\frac{2}{5}}\notag\\
R^p(n)\le \frac{2}{\pi}-\frac{1}{2}c_2^2n^{-4}&\Rightarrow \gamma^p(\omega_r)\ge c'_2|\frac{2}{\pi}-r|^{-\frac{1}{4}}
\end{align}
\yujie for some positive constant $c_d':=(\frac{1}{2}c^2_d)^{\frac{d-1}{d+2}}$ \blk
\end{theorem}
\begin{proof}
For any $n$-outcome POVM $\{ \Pi_{\lambda}\}$, we have a chain of inequality given as:
\begin{equation}
R(n)=\sup_{\{ \Pi_{\lambda}\}}R({\{ \Pi_{\lambda}\}})\le \sup_{\{ \Pi_{\lambda}\}}R({\text{\textbf{sym}}\{ \Pi_{\lambda}\}})=\sup_{\{2\alpha_{\lambda} \hat{n}_{\lambda}\}}\text{inr}(\mathfrak{m}^{*}_{\{ \Pi_{\lambda}\}})=\sup_{\{2\alpha_{\lambda} \hat{n}_{\lambda}\}}\inf_{\hat{x}}\sum_{\lambda}\alpha_{\lambda} |\langle\hat{n}_{\lambda},\hat{x}\rangle|:= R_*(n)
\end{equation}
(see Eq. \eqref{eq:symradius} for the derivation of the last equality).  From Lemma~\ref{prop:Bourgain1989}, we have a lower bound on 
$\norm{\sum_{\lambda}\alpha_{\lambda} |\langle\hat{n}_{\lambda},\hat{x}\rangle|-\beta_d}_{L^2(\sigma_{d})}$. Using Holder's inequality and the fact that $\norm{\sum_{\lambda}\alpha_{\lambda} |\langle\hat{n}_{\lambda},\hat{x}\rangle|-\beta_d}_{L^{\infty}(\sigma_{d})}\le 1$ (\yujie we note this is a quite loose relaxation, but is enough to establish our claim. \blk),
\begin{align}
\norm{\sum_{\lambda}\alpha_{\lambda} |\langle\hat{n}_{\lambda},\hat{x}\rangle|-\beta_d}_{L^2(\sigma_{d})}^2&\le \norm{\sum_{\lambda}\alpha_{\lambda} |\langle\hat{n}_{\lambda},\hat{x}\rangle|-\beta_d}_{L^1(\sigma_{d})}\norm{\sum_{\lambda}\alpha_{\lambda} |\langle\hat{n}_{\lambda},\hat{x}\rangle|-\beta_d}_{L^{\infty}(\sigma_{d})}\notag \\
&\le \norm{\sum_{\lambda}\alpha_{\lambda} |\langle\hat{n}_{\lambda},\hat{x}\rangle|-\beta_d}_{L^1(\sigma_{d})}
\label{eq: holder inequality}
\end{align}
Where $L^1$, and $L^{\infty}$ stand for the $L^1$ and $L^{\infty}$ norm of the integral of measurable function. Additionally, since $\int (\sum_{\lambda}\alpha_{\lambda} |\langle\hat{n}_{\lambda},\hat{x}\rangle|-\beta_d) d\sigma_{d}(\hat{x})=0$, we have
\begin{equation}
\int_{\sum_{\lambda}\alpha_{\lambda} |\langle\hat{n}_{\lambda},\hat{x}\rangle|<\beta_d} \left|\sum_{\lambda}\alpha_{\lambda} |\langle\hat{n}_{\lambda},\hat{x}\rangle|-\beta_d\right|d\sigma_{d}(\hat{x})=\frac{1}{2}\norm{\sum_{\lambda}\alpha_{\lambda} |\langle\hat{n}_{\lambda},\hat{x}\rangle|-\beta_d}_{L^1(\sigma_d)},
\end{equation}
we therefore have
\begin{equation}
\sup_{\hat{x}}\left(\beta_d-\sum_{\lambda}\alpha_{\lambda} |\langle\hat{n}_{\lambda},\hat{x}\rangle|\right)\ge\frac{1}{2}\norm{\sum_{\lambda}\alpha_{\lambda} |\langle\hat{n}_{\lambda},\hat{x}\rangle|-\beta_d}_{L^1(\sigma_d)}\ge  \frac{1}{2}\norm{\sum_{\lambda}\alpha_{\lambda} |\langle\hat{n}_{\lambda},\hat{x}\rangle|-\beta_d}_{L^2(\sigma_d)}^2\ge \frac{1}{2}c_d^2n^{-\frac{d+2}{d-1}}.
\end{equation}
In the end we have
\begin{equation}
 \sup_{\{2\alpha_{\lambda} \hat{n}_{\lambda}\}}
\inf_{\hat{x}}\sum_{\lambda}\alpha_{\lambda} |\langle\hat{n}_{\lambda},\hat{x}\rangle| \le \beta_d-\frac{1}{2}c_d^2n^{-\frac{d+2}{d-1}},
\end{equation}
from which we conclude:
\begin{align}
R(n)\le R_*(n)\le \frac{1}{2}-\frac{1}{2}c_3^2n^{-\frac{5}{2}}&\Rightarrow \gamma(\omega_r)\ge c'_3|\frac{1}{2}-r|^{-\frac{2}{5}}\notag\\
R^p(n)\le R^p_*(n)\le \frac{2}{\pi}-\frac{1}{2}c_2^2n^{-4}&\Rightarrow \gamma^p(\omega_r)\ge c'_2|\frac{2}{\pi}-r|^{-\frac{1}{4}}\notag\\
\end{align}
\yujie for some positive constant $c_d':=(\frac{1}{2}c^2_d)^{\frac{d-1}{d+2}}$. \blk
\end{proof}
\begin{remark}
The bound here for planar measurements $(d=2)$ is less tight than the upper bound we obtained in corollary~\ref{cor:upper planer}.
\end{remark}
\begin{lemma}[\yujie Theorem 1 \cite{siegel2023} \blk]
\label{Prop:lb-Bourgain_SM}
There exists a positive constant $C_d$ and zonotope in $\mbb{R}^d$ with $2n$ 
 generators $\{\pm\alpha_{\lambda} \hat{n}_{\lambda}\}_{\lambda=1}^n$ with $\sum\alpha_{\lambda}=1$ such that 
\begin{equation}
\left|\sum_{\lambda=1}^n\alpha_{\lambda} |\hat{n}_{\lambda}\cdot\hat{x}|-\beta_d\right|<C_dn^{-\frac{d+2}{2(d-1)}}\qquad\forall \hat{x}\in S_{d-1},
\end{equation}
where $\beta_d :=\int_{S_{d-1}}  |\langle\hat{n},\hat{x}\rangle| d\sigma_{d}(\hat{x})$, e.g., $\beta_2=\tfrac{2}{\pi}$ and $\beta_3=\tfrac{1}{2}$.
\end{lemma}
\noindent  Since $\inf_{\hat{x}}\sum_{\lambda=1}^n\alpha_{\lambda} |(\hat{x}\cdot\hat{n}_{\lambda})|\le \sum_{\lambda=1}^n \alpha_{\lambda}\int |(\hat{x}\cdot\hat{n}_{\lambda})|d\sigma_{d}(\hat{x})=\beta_d\sum_{\lambda=1}^n \alpha_{\lambda}=  \beta_d$ by the definition in Eq.~\eqref{eq:defbeta},  from lemma~\ref{Prop:lb-Bourgain_SM},
\begin{equation}
\begin{split}
R(2n)&\ge 
 R(\text{sym}\{ \Pi_{\lambda}\})
= \inf_{\hat{x}}\sum_{\lambda=1}^n\alpha_{\lambda} |(\hat{x}\cdot\hat{n}_{\lambda})|>\beta_d-C_dn^{-\frac{d+2}{2(d-1)}}.
 \end{split}
\end{equation}
Thus we immediately obtain an upper bound on the simulation cost:
\begin{theorem}
For $r<\beta_3=\frac{1}{2}$, the simulation cost is upper bounded by:
\begin{align}
N(r)&\le C'_3|\frac{1}{2}-r|^{-4/5}
\end{align}
for some positive constants $C'_3=2(C_3)^{\frac{4}{5}}$.
\end{theorem}

\section{Criterion for compatible radius $R(\{ \Pi_{\lambda}\})$}
\label{appendixC}
\yujie The following proposition is closely related to theorem 1 in \cite{Nguyen2019} for quantum steering for bipartite state. \blk

\setcounter{proposition}{4}
\begin{proposition}
For a given qubit POVM $\{ \Pi_{\lambda}=\alpha_{\lambda}(\mbb{I}+\eta_{\lambda}\hat{n}_{\lambda}\cdot\vec{\sigma})\}$, the compatible radius $R(\{ \Pi_{\lambda}\})$ is given by:
\begin{equation}
R(\{ \Pi_{\lambda}\})= \inf_{\substack{|\vec{c}|=1\\-1\le c_0\le 1}}\sum_{\lambda}|\alpha_{\lambda} (c_0+\eta_{\lambda} \vec{c}\cdot\hat{n}_{\lambda})|.
\end{equation}
\end{proposition}
\begin{proof}
We start by noticing that the set $\mc{M}_{\{ \Pi_{\lambda}\}}$ is convex for any given $\{ \Pi_{\lambda}\}$.  Therefore, we can always find a set of tight inequalities that bound $\mc{M}_{\{ \Pi_{\lambda}\}}$. Such an inequality can be represented using an operator $C=c_0\mbb{I}+\vec{c}\cdot\vec{\sigma}$. Let $N\in\mc{M}_{\{ \Pi_{\lambda}\}}$, then for any $C$ we can write the inequality as
\begin{equation}
\langle C,N\rangle \le \sup_{{M}\in\mc{M}_{\{ \Pi_{\lambda}\}}} \langle C,M\rangle
\end{equation}
where $\langle X,Y\rangle=\tr[XY]$. 
\par
Given a POVM $N=n_0(\mbb{I}+\vec{n}\cdot\vec{\sigma})$, we can proceed to simplify the above inequality:
\begin{equation}
\begin{split}
    n_0(c_o+\vec{c}\cdot\vec{n})&\le \sum_{\lambda}\sup_{0\le x_{\lambda}\le 1}\alpha_{\lambda} x_{\lambda}(c_0+\eta_{\lambda} \vec{c}\cdot\hat{n}_{\lambda})\\
    &=\sum_{\lambda}\max[\alpha_{\lambda} (c_0+\eta_{\lambda} \vec{c}\cdot\hat{n}_{\lambda}),0],
\end{split}
\end{equation}
where the last equality always holds by setting $x_{\lambda}=0$ whenever $c_0+\eta_{\lambda} \vec{c}\cdot\hat{n}_{\lambda}\le 0$ and $x_{\lambda}=1$ otherwise. Moreover for any unbiased measurement $N\in\mc{M}_{\{ \Pi_{\lambda}\}}$ with $n_0=\frac{1}{2}$ (\yujie which takes noisy spin measurements as a special case\blk), we can simplify the above expression as:
\begin{equation}
\begin{split}
 \vec{c}\cdot\vec{n}\le 2\sum_{\lambda}\max[\alpha_{\lambda} (c_0+\eta_{\lambda} \vec{c}\cdot\hat{n}_{\lambda}),0]-c_0&=2\sum_{\lambda}\max[\alpha_{\lambda} (c_0+\eta_{\lambda} \vec{c}\cdot\hat{n}_{\lambda}),0]-\sum_{\lambda}\alpha_{\lambda} c_0-\sum_{\lambda}\alpha_{\lambda} \eta_{\lambda} \vec{c}\cdot\hat{n}_{\lambda}\\
    &=\sum_{\lambda}|\alpha_{\lambda} (c_0+\eta_{\lambda} \vec{c}\cdot\hat{n}_{\lambda})|,
\end{split}
\end{equation}
where we use  the identity $2\max[t,0]-t=|t|$,  $\sum_{\lambda}\alpha_{\lambda} =1~\text{and}~\sum_{\lambda}\alpha_{\lambda} \eta_{\lambda} \hat{n}_{\lambda}=\vec{0}$ in the first equality. \par
When varying over all choices of operator $C$ (thus, all inequalities that bound the convex set $\mc{M}_{\{ \Pi_{\lambda}\}}$), we finally arrive at the criteria:
\begin{equation}
|\vec{n}|\le \inf_{c_0,\vec{c}}\frac{\sum_{\lambda}|\alpha_{\lambda} (c_0+\eta_{\lambda} \vec{c}\cdot\hat{n}_{\lambda})|}{|\vec{c}|}.
\end{equation}
Since we are always allowed to scale $c_0$ and $|\vec{c}|$, the above infimum can be further simplified with constraint $|\vec{c}|=1$ and $-1\le c_0\le 1$ (the last constraint follows from the fact that minimal of any sum of sign functions $\inf_x\sum_i|x+b_i|$ is achieved with $|x|\le \sup_i|b_i|$):
\begin{equation}
R(\{ \Pi_{\lambda}\})= \inf_{\substack{|\vec{c}|=1\\-1\le c_0\le 1}}\sum_{\lambda}|\alpha_{\lambda} (c_0+\eta_{\lambda} \vec{c}\cdot\hat{n}_{\lambda})|
\end{equation}
The compatibility radius $R(n)$ can then be expressed as
\begin{equation}
R(n)=\sup_{\substack{\{ \Pi_{\lambda}\}}} \inf_{\substack{|\vec{c}|=1\\-1\le c_0\le 1}}\sum_{\lambda}|\alpha_{\lambda} (c_0+\eta_{\lambda} \vec{c}\cdot\hat{n}_{\lambda})|,
\end{equation}
\noindent where the supremum is taken over all $n$-element POVMs. 
\end{proof}
\setcounter{corollary}{2}
\begin{corollary}
    $R(n)$ is maximized by rank-1 POVM $\{\Pi_n\}$, i.e., $\eta_{\lambda} =1$ for $\lambda \in [n]$. 
\end{corollary}
\begin{proof}
Given any POVM $\{\Pi_{\lambda}\}$ with $ \Pi_{\lambda}=\alpha_{\lambda} (\mbb{I}+\eta_{\lambda} \hat{n}_{\lambda}\cdot\vec{\sigma})$, we can define POVM $\{\Pi'_{\lambda}\}$ with:
\begin{equation*}
\Pi'_{\lambda}=\beta_{\lambda}(\mbb{I}+\hat{n}_{\lambda}\cdot\vec{\sigma})
\end{equation*}
where $\beta_{\lambda}=\frac{\alpha_{\lambda} \eta_{\lambda} }{\sum_j\alpha_j\eta_j}$. Let $c^*_0$ and $\vec{c^*}$ be such that
\begin{equation*}
R(\{ \Pi_{\lambda}'\})= \inf_{\substack{|\vec{c}|=1\\-1\le c_0\le 1}}\sum_{\lambda}|\beta_{\lambda}(c_0+\vec{c}\cdot\hat{n}_{\lambda})|=\sum_{\lambda}|\beta_{\lambda}(c^*_0+\vec{c^*}\cdot\hat{n}_{\lambda})|=\frac{1}{\sum_j\alpha_j\eta_j}\sum_{\lambda}|\alpha_{\lambda} \eta_{\lambda} (c^*_0+\vec{c^*}\cdot\hat{n}_{\lambda})|
\end{equation*}
For each individual term above, 
\begin{equation*}
|\alpha_{\lambda} \eta_{\lambda} (c^*_0+\vec{c^*}\cdot\hat{n}_{\lambda})|+|(1-\eta_{\lambda} )\alpha_{\lambda} c_0^*|\ge|\alpha_{\lambda} (c^*_0+\eta_{\lambda} \vec{c^*}\cdot\hat{n}_{\lambda})|.
\end{equation*}
Summing over $\lambda$, we have
\begin{equation*}
\sum_{\lambda}|\alpha_{\lambda} \eta_{\lambda} (c^*_0+\vec{c^*}\cdot\hat{n}_{\lambda})|+\sum_{\lambda}|(1-\eta_{\lambda} )\alpha_{\lambda} c_0^*|\ge \sum_{\lambda}|\alpha_{\lambda} (c^*_0+\eta_{\lambda} \vec{c^*}\cdot\hat{n}_{\lambda})|,
\end{equation*}
which is equivalent to
\begin{equation*}
\sum_j\alpha_j\eta_jR(\{ \Pi_{\lambda}'\})+ (1-\sum_j\alpha_j\eta_j)|c_0^*|\ge \sum_{\lambda}|\alpha_{\lambda} (c^*_0+\eta_{\lambda} \vec{c^*}\cdot\hat{n}_{\lambda})|\ge  R(\{ \Pi_{\lambda}\}).
\end{equation*}
The last inequality holds because $(c^*_0, \vec{c^*})$ is optimal for $\{\Pi'_{\lambda}\}$ but might not be the optimal choice for $\{\Pi_{\lambda}\}$.\par
The last piece of the proof relies on showing $R(\{ \Pi_{\lambda}'\})\ge |c_0^*|$. This holds because:
\begin{equation*}
R(\{ \Pi_{\lambda}'\})=\sum_{\lambda}|\beta_{\lambda}(c^*_0+\vec{c^*}\cdot\hat{n}_{\lambda})|\ge |\sum_{\lambda}\beta_{\lambda}(c^*_0+\vec{c^*}\cdot\hat{n}_{\lambda})|=|c^*_0|,
\end{equation*}
where the last equality holds since $\sum_{\lambda}\Pi_{\lambda}=\mbb{I}$ (i.e., $\sum_{\lambda}\beta_{\lambda}\hat{n}_{\lambda}=\vec{0}$ and $\sum_{\lambda}\beta_{\lambda}=1$ ). Therefore, we finally have: 
\begin{equation*}
R(\{ \Pi_{\lambda}'\})\ge R(\{ \Pi_{\lambda}\}).
\end{equation*}
\end{proof}
\begin{remark}
For symmetric POVM $\text{sym}\{ \Pi_{\lambda}\}$, the compatible radius can be computed as
\begin{align}
R(\text{sym}\{ \Pi_{\lambda}\})= \inf_{\hat{c}}\sum_{\lambda=1}^n\alpha_{\lambda} |(\hat{c}\cdot\eta_{\lambda}\hat{n}_{\lambda})|\label{eq:symradius}
\end{align}
To show this is true for all symmetrically symmetric POVM, notice $\inf_{x}(|x+y|+|x-y|)=2|y|$, thus
\begin{equation}
R(\text{sym}\{ \Pi_{\lambda}\})= \inf_{\substack{|\vec{c}|=1\\-1\le c_0\le 1}}\frac{1}{2}[\sum_{\lambda}|\alpha_{\lambda} (c_0+\eta_{\lambda} \vec{c}\cdot\hat{n}_{\lambda})|+\sum_{\lambda}|\alpha_{\lambda} (c_0-\eta_{\lambda} \vec{c}\cdot\hat{n}_{\lambda})|]=\inf_{\hat{c}}\sum_{\lambda=1}^n\alpha_{\lambda} |(\hat{c}\cdot\eta_{\lambda}\hat{n}_{\lambda})|
\end{equation}
\end{remark}
\yujie In the following we assume $\eta_{\lambda}=1$ for convenience. \blk
\begin{corollary}
The Compatible radius $R(\{ \Pi_{\lambda}\})$ can be computed by considering a finite set of less than $n\choose 3$ operators $C=c_0\mbb{I}+\vec{c}\cdot\vec{\sigma}$, where $(c_0,\vec{c})^T$ is defined by the norm vector of facets of $\mathfrak{M}_{\{ \Pi_{\lambda}\}}$.
\end{corollary}
\begin{proof}
As discussed in Section~\ref{appendixA}, the compatible region is a convex zonotope. Instead of running over all hyperplane boundaries given by operator $C=c_0\mbb{I}+\vec{c}\cdot\vec{\sigma}$, it suffices to consider a finite set of the operator $C$ that define the normal vectors of those facets. A similar idea was discussed in \cite{Nguyen2019}.
\par
From the property of a zonotope, every edge of the compatible region $\mf{M}_{\{ \Pi_{\lambda}\}}$ is always parallel to one of its generator vectors proportional to $(1,\hat{n}_{\lambda})$, while each facet is defined by $d-1=3$ edges of the set. Therefore, we just have to consider at most $\binom{n}{3}$ different choices of normal vectors. To be more specific, the vector $(c_0,\vec{c})$ should be perpendicular to three different vectors $(1,\hat{n}_{\lambda})$.  Hence, we can solve the equations $(c_0,\vec{c})\cdot (1,\hat{n}_{\lambda})^T=0$ and obtain:
\begin{equation}
  \vec{c}=\frac{(\hat{n}_x-\hat{n}_y)\times (\hat{n}_x-\hat{n}_z)}{|(\hat{n}_x-\hat{n}_y)\times (\hat{n}_x-\hat{n}_z)|},~~~~c_0=-\vec{c}\cdot\hat{n}_x
\end{equation}
where $\hat{n}_x, \hat{n}_y$ and $ \hat{n}_z$ are any choices of distinct vectors associated to POVM $\{ \Pi_{\lambda}\}$. Therefore, we can now simplify the criteria above by calculating the minimum value over a finite set of size $\binom{n}{3}$.\par
Similarly, for the case with planar measurements, the vector $(c_0,\vec{c})$ should be perpendicular to two different vectors $(1,\hat{n}_{\lambda})^T$.  Hence we can write:
\begin{equation}
    (c_0,\vec{c})^T={(1,\hat{n}_x)^T\times (1,\hat{n}_y)^T}
\end{equation}
where $\vec{c}, \hat{n}_{\lambda}$ are written as vectors in $\mbb{R}^2$. 
\end{proof}

\section{Example of $R(\{ \Pi_{\lambda}\})$ for symmetric POVMs, and compatible models}
\label{appendixE}
In this section, we compute the compatible radius $R(\{ \Pi_{\lambda}\})$ for different classes of symmetric POVMs, including: \textbf{Result 1}: the compatible radius for qubit rotationally symmetric planar measurements; \\
\textbf{Result 2}: a compatible model based on qubit rotationally symmetric planar measurements; \\
\textbf{Result 3}: the compatible radius for qubit POVMs with regular polyhedron configuration. 

\yujie These results are relatively easy to compute thanks to the remark mentioned in the last section, i.e., for symmetric POVM $\text{sym}\{ \Pi_{\lambda}\}$, the compatible radius can be computed as
\begin{align}
R(\text{sym}\{ \Pi_{\lambda}\})= \inf_{\hat{c}}\sum_{\lambda=1}^n\alpha_{\lambda} |(\hat{c}\cdot\eta_{\lambda}\hat{n}_{\lambda})|
\end{align}
\blk

\begin{result}
For equally spaced planar measurement $\{\Pi^{\rot}_{\lambda}=\frac{1}{n}(\mbb{I}+\hat{n}_{\lambda}\cdot\sigma)\}$ with $\hat{n}_{\lambda}=[\cos(\frac{2\pi \lambda}{n}),0,\sin(\frac{2\pi \lambda}{n})]^T$, we have:
\label{res: planar}
\begin{equation}
R^p(\{\Pi^{\rot}_{\lambda}\})=\begin{cases}
\frac{1}{n}\cot(\frac{\pi}{2n})\cos(\frac{\pi}{2n})~~~~&\text{if $n$ is odd}\\
\frac{2}{n}\cot(\frac{\pi}{n}) &\text{if $n$ is even}.
\end{cases}
\end{equation}
\end{result}
\begin{proof}
If $n$ is even, the original vectors $\hat{n}_{\lambda}$ form a regular polygon with $n$ sides. To get the extreme point $M^n_i\in \mc{M}_{\{ \Pi_{\lambda}\}}$, we can simply check that, they are obtained by adding up half of $\{\Pi^{\rot}_{\lambda}\}$:
\begin{equation}
M^n_i=\sum_{j=i}^{j=n/2+i-1}\Pi^{\rot}_j=\frac{1}{2}(\mbb{I}+\vec{m}_i\cdot\sigma)
\end{equation}
where $\{\vec{m}_i\}$ form exactly the same regular polygon with $|\vec{m}_i|=\frac{2}{n\sin(\frac{\pi}{n})}$. \par
Given an $n$-sided regular polygon, the ratio between its inscribed radius and circumscribed radius is $\frac{r_i}{r_c}=\cos(\frac{\pi}{n})$ (standard regular-polygon formula). Therefore, the circle contained in the $n$-sided polygon defined by $\vec{m}_i$ has radius
\begin{equation}
    r_n=\frac{r_i}{r_c}|\vec{m}_i|=\cos(\frac{\pi}{n})\cdot\frac{2}{n\sin(\frac{\pi}{n})}=\frac{2}{n}\cot(\frac{\pi}{n}).
\end{equation}
If $n$ is odd, the case is slightly different; instead, the extreme points can be enumerated as
\begin{equation}
\begin{split}
M^n_i&=\sum_{j=i}^{j=(n-1)/2+i-1}\Pi^{\rot}_j+\frac{1}{2}\Pi^{\rot}_{(n+1)/2+i-1}=\frac{1}{2}(\mbb{I}+\vec{m}_i\cdot\sigma)\\
W^n_i&=\sum_{j=i}^{j=(n-1)/2+i-1}\Pi^{\rot}_j+\frac{1}{2}\Pi^{\rot}_{i-1}=\frac{1}{2}(\mbb{I}+\vec{w}_i\cdot\vec{\sigma}).
\end{split}
\end{equation}
In total there are $2n$ vectors $\{\vec{m}_i\}\cup\{\vec{w}_i\}$ which forms a $2n$-sided regular polygon with $|\vec{m}_i|=|\vec{w}_i|=\frac{1}{n}\cot(\frac{\pi}{2n})$ and radius
\begin{equation}
    r_n=\frac{r_i}{r_c}|\vec{m}_i|=\cos(\frac{\pi}{2n})\cdot\frac{1}{n}\cot(\frac{\pi}{2n})=\frac{1}{n}\cot(\frac{\pi}{2n})\cos(\frac{\pi}{2n}).
\end{equation}
\end{proof}
From our Table I in the main text, we observe that when $n$  is odd, the rotationally symmetric scheme appears to be optimal. However, for even $n$ , it is consistently suboptimal. This can be quantitatively explained by the fact that the unbiased compatible region $\mathfrak{m}_{\{ \Pi_{\lambda}\}}$ is an $n$-sided polygon for even $n$, but a $2n$-sided polygon for odd $n$.  More specifically, when $n$  is even, there is an inherent trade-off between maximizing the number of sides of $\mathfrak{m}_{\{ \Pi_{\lambda}\}}$ and maintaining its symmetry. The rotationally symmetric scheme exhibits the highest degree of symmetry; however, in this case, the unbiased compatible region has only  $n$  sides. In contrast, by adopting a slightly asymmetric POVM, we can construct a $2n$-sided $\mathfrak{m}_{\{ \Pi_{\lambda}\}}$, which, while not a regular polygon, results in a larger inradius.
\begin{result}
Here we provide explicit compatible model (equivalently a local hidden state model) when the parent POVM is rotationally symmetric.
\end{result}
For finite $n$, with parent POVM $\Pi^{\rot}_{\lambda}=\frac{1}{n}(\mbb{I}+\hat{n}_{\lambda}\cdot\vec{\sigma})$, any arbitrary child POVM $T=\frac{1}{2}(\mbb{I}+R(\{\Pi^{\rot}_{\lambda}\})\hat{t}\cdot\vec{\sigma})$ with $\hat{t}=(\cos(\theta),0,\sin(\theta))^T$ can be simulated as 
$$T=\sum_{\lambda}p_{\lambda}\Pi^{\rot}_{\lambda}$$
{For even $n$:~~} 
$$ p_{\lambda}=\frac{1+\frac{\text{sign}(\hat{m}_1\cdot\hat{n}_{\lambda})\sin(\frac{\pi}{n}-x)+\text{sign}(\hat{m}_2\cdot\hat{n}_{\lambda})\sin(\frac{\pi}{n}+x)}{2\sin(\frac{\pi}{n})}}{2}$$\par

with $x=\theta +\frac{2\pi k}{n}\in [-\frac{\pi}{n},\frac{\pi}{n}]$,  $\hat{m}_1=[\cos(\frac{2\pi (k-\frac{1}{2})}{n}),0,\sin(\frac{2\pi (k-\frac{1}{2})}{n})]^T$, $\hat{m}_2=[\cos(\frac{2\pi (k+\frac{1}{2})}{n}),0,\sin(\frac{2\pi (k+\frac{1}{2})}{n})]^T$ and $k\in \mbb{Z}$.\\
\\
{For odd $n$:~~} 
\begin{align*}
    p_{\lambda}=\frac{1+\frac{\text{sign}(\hat{m}_1\cdot\hat{n}_{\lambda})\sin(\frac{\pi}{2n}-x)+\text{sign}(\hat{m}_2\cdot\hat{n}_{\lambda})\sin(\frac{\pi}{2n}+x)}{2\sin(\frac{\pi}{2n})}}{2}
\end{align*}\par
with $x=\theta +\frac{\pi k}{n}\in [-\frac{\pi}{2n},\frac{\pi}{2n}]$,  $\hat{m}_1=[\cos(\frac{\pi (k-\frac{1}{2})}{n}),0,\sin(\frac{\pi (k-\frac{1}{2})}{n})]^T$, $\hat{m}_2=[\cos(\frac{\pi (k+\frac{1}{2})}{n}),0,\sin(\frac{\pi (k+\frac{1}{2})}{n})]^T$ and $k\in \mbb{Z}$.\\
When taking $n\rightarrow \infty$, we have $\hat{m}_1=\hat{m}_2=\hat{t}$ and $\sin(f(x))\rightarrow f(x)$ when $f(x)\rightarrow 0$. The above model becomes:
\begin{equation}
 p_{\lambda}=\frac{1+\frac{\text{sign}(\hat{m}_1\cdot\hat{n}_{\lambda})\sin(\frac{\pi}{n}-x)+\text{sign}(\hat{m}_2\cdot\hat{n}_{\lambda})\sin(\frac{\pi}{n}+x)}{2\sin(\frac{\pi}{n})}}{2}\longrightarrow \frac{1+\frac{\text{sign}(\hat{t}\cdot\hat{n}_{\lambda})(\frac{2\pi}{n}-x)+\text{sign}(\hat{t}\cdot\hat{n}_{\lambda})x}{2\frac{\pi}{n}}}{2}=\frac{1+\text{sign}(\hat{t}\cdot\hat{n}_{\lambda})}{2}   
\end{equation}
Letting $\lim_{n\rightarrow\infty}\sum_{i=1}^n(\cdot)\frac{1}{n}=\int_{0}^{2\pi}(\cdot)\frac{d\psi}{2\pi}$, we can define $\Pi_{\psi}=\frac{1}{2\pi}(\mbb{I}+\hat{n}_{\psi}\cdot\vec{\sigma})$ and $p_{\psi}=\frac{1+\text{sign}(\hat{t}\cdot\hat{n}_{\psi})}{2}$ from the discrete one above to a continuous one, where $\hat{n}_{\psi}=(\cos(\psi),0,\sin(\psi))^T$. Therefore, any child POVM $T=\frac{1}{2}(\mbb{I}+\frac{2}{\pi}\hat{t}\cdot\vec{\sigma})$ can be written as:
\begin{equation}
T=\int_0^{2\pi} d\theta p_{\psi} \Pi_{\phi}=\int_0^{2\pi} d\psi \left(\frac{1+\text{sign}(\hat{t}\cdot\hat{n}_{\psi})}{2}\right) \frac{1}{2\pi}(\mbb{I}+\hat{n}_{\psi}\cdot\vec{\sigma})=\frac{1}{2}(\mbb{I}+\frac{2}{\pi}\hat{t}\cdot\vec{\sigma})
\end{equation}
\\
\\
\begin{result}
For a parent POVM $\{ \Pi_{\lambda}\}$ with platonic-solid configuration, we can compute the $R(\{ \Pi_{\lambda}\})$ similarly using Result~\ref{res: planar}. The polyhedron $\mc{M}_{\{ \Pi_{\lambda}\}}$ associated with different platonic configurations are summarized below:
\begin{table}[h]
    \centering
    \begin{tabular}{|c|c|c|c|}
    \hline
Complexity & POVM $\{ \Pi_{\lambda}\}$ & Compatible region $\msf{m}_{\{ \Pi_{\lambda}\}}$ & compatible radius $R(\{ \Pi_{\lambda}\})$\\
\hline
4 &  tetrahedron & Octahedron &$\frac{1}{3}$\\
\hline
6 &   Octahedron & Cube & $\frac{1}{3}$\\
\hline
8 &   Cube & Rhombic dodecahedron &$\frac{\sqrt{6}}{6}$\\
\hline
12 &  Icosahedron & Rhombic triacontahedron & $\frac{\phi^3\sqrt{1+(1-\phi)^2}}{3({1+\phi^2})}$\\ 
\hline
 20 &  Dodecahedron &  Rhombic enneacontahedron & $\sqrt{\frac{5}{6}}\frac{\phi^2}{5}$\\
 \hline
    \end{tabular}
    \label{tab:my_label}
    \caption{Parent POVM with platonic configuration and their associated unbiased compatible region $\msf{m}_{\{ \Pi_{\lambda}\}}$, where $\phi=\frac{1+\sqrt{5}}{2}$ is the golden ratio.}
\end{table}\\
\begin{figure}[t]
    \centering
    \includegraphics[width=0.95\textwidth]{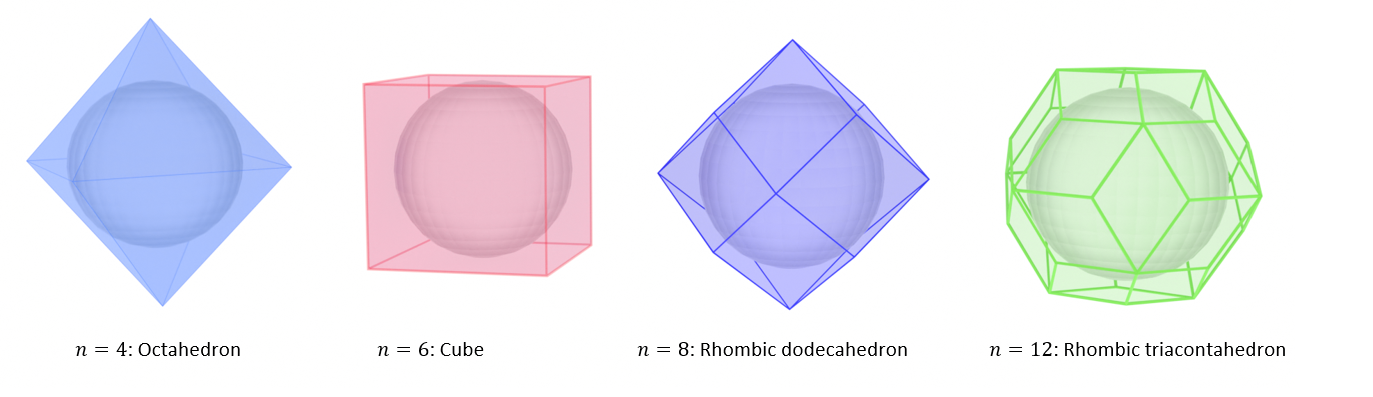}
    \caption{Compatible region for Platonic solid with $n=4,6,8,12$. From left to right: Octahedron, Cube, Rhombic dodecahedron and Rhombic triacontahedron .}
    \label{fig:planar}
\end{figure}
\par 
\textbf{Tetrahedron}:
With the vertices of an Tetrahedron, we can construct a 4-element parent POVM as follows:
\begin{equation}
{\Pi}^T_{\lambda}=\frac{1}{4}(\mbb{I}+\hat{n}_{\lambda}\cdot\vec{\sigma})~~~~\text{with}~~ \hat{n}_{\lambda}\in\left\{\frac{1}{\sqrt{3}}(1,1,1),\frac{1}{\sqrt{3}}(1,-1,-1),\frac{1}{\sqrt{3}}(-1,1,-1),\frac{1}{\sqrt{3}}(-1,-1,1)\right\}\\
\end{equation}
The extreme points of the associated compatible region $\msf{m}_{\{ \Pi_{\lambda}\}}$ can be computed as:
\begin{equation}
\begin{split}
{M}^T_{\lambda}&=\frac{1}{2}(\mbb{I}+\vec{m}_{\lambda}\cdot\vec{\sigma})~~~~\text{with}~~ \vec{m}_{\lambda}\in\left\{\frac{1}{\sqrt{3}}(\pm 1,0,0),\frac{1}{\sqrt{3}}(0,\pm 1,0),\frac{1}{\sqrt{3}}(0,0,\pm 1)\right\}\\
\end{split}
\end{equation}
These 6 vertices altogether form a convex polyhedron -- Octahedron, to calculate the inscribed radius, we have:
 $$r^T=|\vec{m}_{\lambda}|\cdot\frac{1}{\sqrt{3}}=\frac{1}{3}$$
\par
\textbf{Octahedron}:
With the vertices of a Octahedron, we can construct a 6-element parent POVM as follows:
\begin{equation}
\Pi^O_{\lambda}=\frac{1}{6}(\mbb{I}+\hat{n}_{\lambda}\cdot\vec{\sigma})~~~~\vec{n}_{\lambda}\in\{(\pm1,0,0),(0,\pm1,0),(0,0,\pm1)\},
\end{equation}
The extreme points of the associate compatible region $\msf{m}_{\{ \Pi_{\lambda}^O\}}$ can be computed as:
\begin{equation}
\begin{split}
{M}^O_{\lambda}&=\frac{1}{2}(\mbb{I}+\vec{m}_{\lambda}\cdot\sigma)~~~~\vec{m}_{\lambda}\in\{\frac{1}{3}(\pm1,\pm1,\pm1)\},
\end{split}
\end{equation}
These 8 vertices altogether form a convex polyhedron -- Cube, to calculate the inscribed radius, we have:
 $$r^O=|\vec{m}_{\lambda}|\cdot\frac{1}{\sqrt{3}}=\frac{1}{3}$$
\par
\textbf{Cube}:
With the vertices of a Cube, we can construct an 8-element parent POVM as follows:
\begin{equation}
{\Pi}^C_{\lambda}=\frac{1}{8}(\mbb{I}+\hat{n}_{\lambda}\cdot\vec{\sigma})~~~~\vec{n}_{\lambda}\in\{\frac{1}{\sqrt{3}}(\pm1,\pm1,\pm1)\}\\
\end{equation}
The extreme point of the associated compatible region $\msf{m}_{\{ \Pi_{\lambda}^C\}}$ can be computed as:
\begin{equation}
{M}^C_{\lambda}=\frac{1}{2}(\mbb{I}+\vec{m}_{\lambda}\cdot\sigma)~~~~\text{with}~~\vec{m}_{\lambda}\in\{\frac{1}{2\sqrt{3}}(\pm1,\pm1,\pm1),\frac{1}{2\sqrt{3}}(\pm2,0,0),\frac{1}{2\sqrt{3}}(0,\pm2,0),\frac{1}{2\sqrt{3}}(0,0,\pm2)\}
\end{equation}
These 14 vertices altogether form a convex polyhedron --{Rhombic dodecahedron} with edge length $a=|\vec{m}_{\lambda}-\vec{m}_j|=\frac{1}{2}$
 $$r^C=\frac{\sqrt{6}}{3}\cdot\frac{1}{2}=\frac{\sqrt{6}}{6}\approx0.408$$
\par
\textbf{Icosahedron}:
With the vertices of an icosahedron, we can construct a 12-element parent POVM as following:
\begin{equation}
{\Pi}^I_{\lambda}=\frac{1}{12}(\mbb{I}+\hat{n}_{\lambda}\cdot\vec{\sigma})~~~~\text{with}~~ \hat{n}_{\lambda}\in\left\{\frac{1}{\sqrt{1+\phi^2}}(0,\pm1,\pm\phi),\frac{1}{\sqrt{1+\phi^2}}(\pm\phi,0,\pm1),\frac{1}{\sqrt{1+\phi^2}}(\pm1,\pm\phi,0)\right\}\\
\end{equation}
where $\phi=\frac{1+\sqrt{5}}{2}$.  This time the extreme points of the associate compatible region $\msf{m}_{\{ \Pi_{\lambda}\}}$ can be computed as:
\begin{equation}
\begin{split}
{M}^I_{\lambda}&=\frac{1}{2}(\mbb{I}+\hat{m}_{\lambda}\cdot\vec{\sigma})~~~~\text{with}~~ \hat{m}_i\in\left\{{\frac{\phi}{3\sqrt{1+\phi^2}}}(0,\pm1,\pm\phi),{\frac{\phi}{3\sqrt{1+\phi^2}}}(\pm\phi,0,\pm1),{\frac{\phi}{3\sqrt{1+\phi^2}}}(\pm1,\pm\phi,0)\right\}\\
&\bigcup\left\{\frac{\phi}{3\sqrt{1+\phi^2}}(\pm1,\pm1,\pm1),\frac{\phi}{3\sqrt{1+\phi^2}}(0,\pm\phi,\pm1/\phi),\frac{\phi}{3\sqrt{1+\phi^2}}(\pm1/\phi,0,\pm\phi),\frac{\phi}{3\sqrt{1+\phi^2}}(\pm\phi,\pm1/\phi,0)\right\}\\
\end{split}
\end{equation}
These 32 vertices altogether form a convex polyhedron -- Rhombic triacontahedron, to calculate the inscribed radius, we first notice that the edge length of Rhombic triacontahedron given as $a=|\vec{m}_{\lambda}-\vec{m}_j|=\frac{\phi}{3\sqrt{1+\phi^2}}\sqrt{1+(1-\phi)^2}$, and the inscribed radius is $r=\frac{\phi^2}{\sqrt{1+\phi^2}}a$. Therefore, the radius of the shrinking Bloch sphere built by parent POVM $\{\Pi^I_i\}$ will be $$r^I=\frac{\phi}{3\sqrt{1+\phi^2}}\sqrt{1+(1-\phi)^2}\cdot\frac{\phi^2}{\sqrt{1+\phi^2}}\approx0.4588$$
\par 
\textbf{Dodecahedron:} With the vertices of a dodecahedron, we can construct a 20-element parent POVM as follows:
\begin{equation}
\begin{split}
{\Pi}^D_{\lambda}&=\frac{1}{20}(\mbb{I}+\hat{n}_{\lambda}\cdot\vec{\sigma})~~~~\text{with}~~ \hat{n}_{\lambda}\in\left\{\frac{1}{\sqrt{3}}(\pm1,\pm1,\pm1)\right\}\\
&\text{or with}~~ \hat{n}_{\lambda}\in\left\{\frac{1}{\sqrt{1+\phi^4}}(0,\pm1,\pm\phi^2),\frac{1}{\sqrt{1+\phi^4}}(\pm\phi^2,0,\pm1),\frac{1}{\sqrt{1+\phi^4}}(\pm1,\pm\phi^2,0)\right\}\\
\end{split}
\end{equation}
The extreme points of the associate compatible region $\msf{m}_{\{ \Pi_{\lambda}^D\}}$ can be computed as:
\begin{equation}
\begin{split}
{M}^D_{\lambda}&=\frac{1}{2}(\mbb{I}+\hat{m}_{\lambda}\cdot\vec{\sigma})~~~~\text{with}~~ \hat{m}_i\in\left\{{\frac{2(\phi^2+1)}{10\sqrt{3}\phi}}(0,\pm\phi,\pm1),{\frac{2(\phi^2+1)}{10\sqrt{3}\phi}}(\pm\phi,\pm1,0),{\frac{2(\phi^2+1)}{10\sqrt{3}\phi}}(0\pm1,0,\pm\phi)\right\}\\
&\bigcup\left\{\frac{2\phi^2}{10\sqrt{3}}(\pm1,\pm1,\pm1),\frac{2\phi^2}{10\sqrt{3}}(0,\pm1/\phi,\pm\phi),\frac{2\phi^2}{10\sqrt{3}}(\pm\phi,0,\pm1/\phi),\frac{2\phi^2}{10\sqrt{3}}(\pm1/\phi,\pm\phi,0)\right\}\\
&\bigcup\left\{\frac{2\phi^2}{10\sqrt{3}}(\pm2/\phi,\pm1/\phi^2,\pm1),\frac{2\phi^2}{10\sqrt{3}}(\pm1/\phi^2,\pm1,\pm2/\phi),\frac{2\phi^2}{10\sqrt{3}}(\pm1,\pm2/\phi,\pm1/\phi^2)\right\}\\
&\bigcup\left\{\frac{2\phi^2}{10\sqrt{3}}(\pm\phi,\pm1/\phi,\pm1/\phi),\frac{2\phi^2}{10\sqrt{3}}(\pm1/\phi,\pm\phi,\pm1/\phi),\frac{2\phi^2}{10\sqrt{3}}(\pm1/\phi,\pm1/\phi,\pm\phi)\right\}\\
&\bigcup\left\{{\frac{1}{10\sqrt{3}}}(\pm2(\phi^2+1),\pm2(\phi-1),0),{\frac{1}{10\sqrt{3}}}(\pm2(\phi-1),\pm2(\phi^2+1),0),{\frac{1}{10\sqrt{3}}}(0,\pm2(\phi^2+1),\pm2(\phi-1))\right\}
\end{split}
\end{equation}
These 92 vertices altogether form a convex polyhedron -- Rhombic enneacontahedron, the inradius of which can be computed as: 
$$r^D=\sqrt{\frac{5}{6}}\frac{\phi^2}{5}=0.4780$$
\end{result}

\section{Simulation cost of entangled state and compatible measurements}
\label{appendixF}
\setcounter{lemma}{0}
\yujie 
It is well known that a noisy two-qubit Werner state is unsteerable under projective measurements (PVMs) iff the corresponding family of noisy qubit spin observables with the same visibility is jointly measurable \cite{Uola2014,Quintino2014}. We establish an analogous correspondence here: (i) an LHS model for the noisy two-qubit Werner state under PVMs is equivalent to (ii) a compatibility model (parent POVM plus classical post-processing) for all noisy qubit spin measurements with the same noise parameter, under a fixed shared randomness budget. The proof is similar to that of \cite{Uola2014,Quintino2014}; we include it below for completeness. 
\blk
\begin{lemma}
$\mc{P}_{r}$ can be simulated with an $n$-element POVM if and only if the simulation cost of Werner states $\omega_r$ under \textit{projective measurements} equals $n$.
\end{lemma}
\yujie
\begin{proof}
Assume $\mc P_r$ is simulated by an $n$-outcome parent POVM $\{\Pi_\lambda\}_{\lambda=1}^n$ with classical post-processing, then any noisy spin measurement $M_{\pm|\hat{n}}^{(r)}=rM_{\pm|\hat{n}}+\frac{1}{2}(1-r)\mbb{I}$ can be simulated as:
\begin{equation}
    M_{\pm|\hat{n}}^{(r)}=\sum_{\lambda=1}^n p(\pm|\hat{n},\lambda)\,\Pi_\lambda, 
\quad \sum_{\pm} p(\pm|\hat{n},\lambda)=1.
\end{equation}
Therefore, for any noisy Werner state $\omega_r=r\op{\Psi^-}{\Psi^-}+\frac{1}{4}(1-r)\mbb{I}\otimes\mbb{I}$ the assemblage on $B$ prepared with noisy spin measurement is
\begin{equation}
\sigma^{(r)}_{\pm|\hat{n}}=\tr_A\big[(M_{\pm|\hat{n}}\otimes\mbb I)\omega_r\big]=\tr_A\!\big[(M_{\pm|\hat{n}}^{(r)}\otimes\mbb I)\op{\Psi^-}{\Psi^-}\big]
=\sum_{\lambda=1}^n p(\pm|\hat{n},\lambda)\,
\tr_A\!\big[(\Pi_\lambda\otimes\mbb I)\op{\Psi^-}{\Psi^-}\big].
\end{equation}
Let $\rho_\lambda:=\tr_A[(\Pi_\lambda\otimes\mbb I)\op{\Psi^-}{\Psi^-}]$.  
Then $\sigma^{(r)}_{\pm|\hat{n}}=\sum_{\lambda}p(\pm|\hat{n},\lambda)\rho_\lambda$, so $\{\rho_\lambda\}$ is a valid LHS ensemble.
Hence the Werner assemblage admits an LHS model with $n$ hidden states (shared randomness cost $\log_2 n$ bits).

Conversely, suppose the Werner assemblage admits an LHS model with $n$ hidden states:
$\sigma^{(r)}_{\pm|\hat{n}}=\sum_{\lambda} p(\pm|\hat{n},\lambda)\,\rho_\lambda$, with
$\rho_\lambda\ge0$ and $\sum_{\lambda}\rho_\lambda=\tfrac12\mbb I$.
From the identity $\tr_A\big[(X\otimes\mbb I)\,\op{\Psi^-}{\Psi^-}\big]
=\frac{1}{2}\sigma_yX^{T}\sigma_y$, we have
\begin{equation}
\frac{1}{2}\sigma_y\big(M_{\pm|\hat{n}}^{(r)}\big)^{T}\sigma_y
=\sum_{\lambda} p(\pm|\hat{n},\lambda)\,\rho_\lambda \Rightarrow
\frac{1}{2}M_{\pm|\hat{n}}^{(r)}
=\sum_{\lambda} p(\pm|\hat{n},\lambda)\big(\sigma_y\rho_\lambda\sigma_y\big)^{T}.
\end{equation}
Therefore, we can define $\Pi_\lambda:=2\,\big(\sigma_y\rho_\lambda\sigma_y\big)^{T}\ge0$, which satisfies
\[
M_{\pm|\hat{n}}^{(r)}=\sum_{\lambda} p(\pm|\hat{n},\lambda)\Pi_\lambda,
\qquad
\sum_{\lambda}\Pi_\lambda
=2\Big(\sum_{\lambda}\sigma_y\,\rho_\lambda\,\sigma_y\Big)^{\sf T}
=\mbb I,
\]
so $\{\Pi_\lambda\}$ is an $n$-outcome parent POVM that simulates $\mc P_r$.

\smallskip
Hence the minimal number of parent outcomes needed to simulate $\mc P_r$ equals the minimal number of hidden states in an LHS model for $\omega_r$ under PVMs; equivalently, both simulation costs are $n$ (or $\log_2 n$ shared bits).
\end{proof}

\blk

\section{Discussion on prop.~\ref{prop:general CC}, prop.~\ref{prop:planar CC} and conj.~\ref{conj:POVM vs PVM}} 
\label{appendixG}
\subsection{Inequivalence between simulation cost of PVMs and POVMs}
In an earlier paper\cite{Zhang2024, Renner2024}, the equivalence of POVMs and PVMs in the context of quantum steering was shown, which is by establishing that a compatible model exists for all noisy POVMs whenever such a model exists for all noisy PVMs at the same noise threshold. However, this elegant equivalence disappears when considering scenarios that involve only finite-shared randomness and specific parent POVMs (or specific local hidden states).\par 

In the main text, we conduct an extensive discussion on the feasibility of a compatible model for noisy PVMs with finite shared randomness. Notably, one intriguing observation we make there is that the optimal parent POVM for simulating the entire set of noisy PVMs may not necessarily be centrally symmetric, even if it does exist. The realization that asymmetry has alue in simulating all PVMs provides insight into the potential divergence between POVMs and PVMs within this restricted compatible model.\par 

In the following sections, we provide a rigorous proof for the aforementioned assertion using Farka’s lemma. \yujie 
To cast the simulation problem as a feasibility test, represent each rank-1 qubit parent effect
\(\Pi_{\lambda}=\alpha_{\lambda}\left(\mathbb{I}+\hat n_{\lambda}\cdot\vec\sigma\right)\)
by the Pauli-basis 4-vector
\(\vec\pi_{\lambda}=(\alpha_{\lambda},\alpha_{\lambda}\hat n_{\lambda})^{\top}\in\mathbb{R}^4\),
and each target (child) effect
\(M_a=\mu_a\left(\mathbb{I}+r\hat m_a\cdot\vec\sigma\right)\)
by
\(\vec m_a=\mu_a(1, r\hat m_a)^{\top}\).
Joint simulability with post-processing coefficients \(p_{a|\lambda}\ge 0\) is exactly the set of linear constraints
\[
\sum_{\lambda=1}^n p_{a|\lambda}\,\vec\pi_{\lambda}=\vec m_a\quad(\forall a),\qquad
\sum_{a=1}^m p_{a|\lambda}=1\quad(\forall \lambda).
\]
Stacking all outcomes yields a single system \(A\mathbf p=\mathbf b\) with
\(A\in\mathbb{R}^{(4m+n)\times (mn)}\) formed by \(m\) blocks
\([\vec\pi_1\ \cdots\ \vec\pi_n]\) (one per \(a\)) and \(n\) normalization rows,
\(\mathbf p=\big(p_{1|1},\ldots,p_{m|1},\,p_{1|2},\ldots,p_{m|n}\big)^{\top}\in\mathbb{R}^{mn}\),
and \(\mathbf b=(\vec m_1,\ldots,\vec m_m,1,\ldots,1)^{\top}\in\mathbb{R}^{4m+n}\). Specifically: \blk

\setcounter{MaxMatrixCols}{20}
\begin{align}
    A&=\begin{pmatrix} \vec{\pi}_1&\cdots&\vec{\pi}_n&0&\cdots&0&0&0&0\\
& \ddots &&& \ddots &&& \ddots &\\ 0&0&0&0&\cdots&0&\vec{\pi}_1&\cdots&\vec{\pi}_n\\
    1&\cdots  &0&1&\cdots    &0&1&\cdots  &0\\
    &\ddots  &&&\ddots  &&&\ddots  &\\
    0&\cdots  &1&0&\cdots  &1& 0&\cdots  &1\end{pmatrix},&    
    \boldsymbol{p}&=\begin{pmatrix}p_{1|1}\\ \vdots \\p_{1|n} \\p_{2|n}\\ \vdots\\ p_{m|n}\end{pmatrix},& \boldsymbol{b}&=\begin{pmatrix}\vec{m}_1\\\vdots \\\vec{m}_m\\1\\ \vdots\\ 1\end{pmatrix}.
\end{align}
\begin{lemma}[Farkas' Lemma]
\begin{equation}
  \nexists \boldsymbol{p}\ge \boldsymbol{0} \text{ s.t. } A\cdot \boldsymbol{p}=\boldsymbol{b} \Longleftrightarrow \exists y \text{ s.t. } A^T\cdot \boldsymbol{y}\ge \boldsymbol{0} \text{ and }\boldsymbol{b}^T\cdot \boldsymbol{y}<0
\end{equation}
where $\boldsymbol{y}\in\mbb{R}^{4m+n}$ as $y=(\vec{y}_1,\cdots,\vec{y}_m,z_1,\cdots,z_n)^T$. 
\end{lemma}
Therefore, to demonstrate the non-existence of a positive, normalized response function $p_{a|i}$ for simulating $\{\wt{M}_a\}$ with $\{\Pi_{\lambda}\}_{\lambda=1}^n$, it suffices to find a vector $\boldsymbol{y}$ such  that $A^T\cdot \boldsymbol{y}\ge \boldsymbol{0} \text{ and }\boldsymbol{b}^T\cdot \boldsymbol{y}<0$ holds at the same time. 
\begin{proposition}
There exists a five-effect parent POVM $\{\Pi_{\lambda}\}_{\lambda=1}^5$ with $R(\{\Pi_{\lambda}\})\approx 0.3714$, that is, it can simulate all noisy spin measurements with radius $r\le 0.3714$, whereas a compatible model fails to exist for some three-outcome noisy POVMs with radius $r> 0.3220$
\end{proposition}
\begin{proof}
The 5-effect POVM $\{\Pi_{\lambda}=\alpha_{\lambda}(\mbb{I}+\hat{n}_{\lambda}\cdot \sigma)\}$ to be considered is of the form of:
\yujie
\begin{align}
\vec{\pi}_1&=0.220(1,\hat{n}_3)^T~~~~\hat{n}_1=(-\frac{12}{55}, \sqrt{1-(\frac{12}{55})^2},0)^T \notag\\
\vec{\pi}_2&=0.220(1,\hat{n}_4)^T~~~~\hat{n}_2=(-\frac{12}{55}, -\frac{1}{2}\sqrt{1-(\frac{12}{55})^2},\frac{\sqrt{3}}{2}\sqrt{1-(\frac{12}{55})^2} )^T \notag\\
\vec{\pi}_3&=0.220(1,\hat{n}_5)^T~~~~\hat{n}_3=(-\frac{12}{55}, -\frac{1}{2}\sqrt{1-(\frac{12}{55})^2},-\frac{\sqrt{3}}{2}\sqrt{1-(\frac{12}{55})^2} )^T \notag \\
\vec{\pi}_4&=0.242(1,\hat{n}_1)^T~~~~\hat{n}_4=(1,0,0)^T \notag\\
\vec{\pi}_5&=0.098(1,\hat{n}_2)^T~~~~\hat{n}_5=(-1, 0, 0)^T 
\label{eq:5-parent}
\end{align}
\blk
First, we used the criteria for computing the compatible radius $R(\{\Pi_{\lambda}\})$ in appendix~\ref{appendixC} to obtain a close form expression: 
\begin{equation}
R(\{\Pi_{\lambda}\})=0.34+0.144\frac{12}{55}\approx 0.3714.
\end{equation}
Therefore, any PVMs with visibility $r\le 0.3714$ can be simulated by the corresponding $5$-outcome POVM. \par 
Now, we give examples showing the infeasibility of simulating three-outcome POVMs $\{M_a^r\}$ with this five-effect parent POVM for an even smaller threshold $r<R(\{\Pi_{\lambda}\})$.  The specific three-outcome POVM $\{M_a^r\}$ we consider is given by the four-vectors
\begin{align}
\vec{m}_1&=\frac{1}{3}(1,r\hat{m}_1)^T~~~~\hat{m}_1=(0,-1,0)^T \notag \\
\vec{m}_2&=\frac{1}{3}(1,r\hat{m}_2)^T~~~~\hat{m}_2=(0,\frac{1}{2}, \frac{\sqrt{3}}{2})^T \notag\\
\vec{m}_3&=\frac{1}{3}(1,r\hat{m}_3)^T~~~~\hat{m}_3=(0,\frac{1}{2}, -\frac{\sqrt{3}}{2})^T.
\label{eq:counter-3-outcome}
\end{align}
By considering vector $\boldsymbol{y}\in\mbb{R}^{4\times 3+5}$ as $y=(\vec{y}_1,\vec{y}_2,\vec{y}_3,z_1,z_2,z_3,z_4,z_5)^T$, where:
\begin{align}
\vec{y}_1&=(0,-\hat{m}_1)^T,~~~~\vec{y}_2=(0,-\hat{m}_2)^T~~~~\vec{y}_3=(0,-\hat{m}_3)^T\notag \\
z_1&=z_2=z_3=0.110\sqrt{1-(\frac{12}{55})^2},~~~~z_4=z_5=0,
\end{align}
and $A^T\cdot \boldsymbol{y}\ge 0$, one can easily verify that with $ \boldsymbol{b}=(\vec{m}_1,\vec{m}_2,\vec{m}_3,1,1,1,1,1)^T$:
\begin{equation}
   \boldsymbol{b}^T\cdot \boldsymbol{y}=-r+0.330\sqrt{1-(\frac{12}{55})^2}<0\Longrightarrow r>0.330\sqrt{1-(\frac{12}{55})^2}\approx 0.3220
\end{equation}
Therefore, from farkas' lemma, for $r>0.3220$, the three-outcome POVM defined in Eq.~\eqref{eq:counter-3-outcome} can never be simulated by the five-effect POVM $\{\Pi_{\lambda}\}$, whereas all PVMs can be simulated by it with $r<0.3714$.

\end{proof}
\setcounter{conjecture}{0}
\begin{conjecture}
The $5$-outcome POVM in Eq.~\eqref{eq:5-parent} is near optimal in terms of simulating PVMs, i.e, close to the optimal $\{\Pi_{\lambda}\}$ that has the largest compatible radius $R(\{\Pi_{\lambda}\}) \approx 0.3718$ reported in the main text.  Therefore, we conjecture that the set of noisy POVMs and noisy PVMs has different simulation costs. 
\end{conjecture}
\subsection{Equivalence between simulation cost of PVMs and POVMs for qubit planar measurements}
However, as noted in Proposition~\ref{prop:planar CC} in the main text, there exist special cases where the simulation cost of a set of noisy qubit PVMs and a set of noisy qubit POVMs is actually the same.

To continue our discussion, we define a few quantities generalized from the main text:
\\ (1) $N^p_{\text{POVM}}(r)$:  the simulation cost of noisy planar qubit POVMs with visibility $r$ ;
\\ (2)  $R^p_{\text{POVM}}(n)$ he compatibility radius for simulating noisy planar qubit POVMs with $n$-outcome measurements. \par 

The following lemma, which will be critical for the subsequent proposition, is stated and proved below:
\begin{lemma}
$R^p_{\text{POVM}}(3)=R^p(3)=\frac{1}{2}\Leftrightarrow  \gamma^p(\omega_{\frac{1}{2}})=N^p(\frac{1}{2})=3$
\label{lem: n=3}
\end{lemma}
\begin{proof}
We first notice that, $R^p_{\text{POVM}}(3)$ has to be obtained by $3$-outcome parent $\{\Pi_{\lambda}\}$ with linear independent effects since $R^p(3)=0$ otherwise.  \par 
$R^p_{\text{POVM}}(3)\le R^p(3)$ is obvious since PVMs form a subset of POVMs.\par  To show $R^p_{\text{POVM}}(3)\ge R^p(3)$, for planar POVM $\{M^r_a\}$, it suffices to consider three outcome measurements since they are extreme in the planar case,  given a fixed noise threshold $r$, if all noisy $3$-outcome qubit planar measurement can be simulated, then all noisy qubit planar POVMs can be simulated. Each of its effects $M_a^r$ can be simulated individually by $\{\Pi_{\lambda}\}$ for all $r\le R^p(3)$, therefore there exist $0\le p_{a|\lambda}\le 1$ such that $M^r_a=\sum_{\lambda=1}^3p_{a|\lambda}\Pi_{\lambda}$.\par 
Now considering the two normalization condition: $\sum_{a=1}^3M^r_a=\sum_{\lambda=1}^3\Pi_{\lambda}=\mbb{I}$, we have:
\begin{equation}
\sum_{\lambda=1}^3 (1-\sum_ap_{a|\lambda})\Pi_{\lambda}=0
\end{equation}
The linearly independence of $\{\Pi_{\lambda}\}$ implies $\sum_ap_{a|\lambda}=1$, thus $\{M^r_a\}_a$ can be simultaneous simulated by $\{\Pi_{\lambda}\}$, and therefore, $R^p_{\text{POVM}}(3)\ge R^p(3)$. We thus have $R^p_{\text{POVM}}(3)=R^p(3)=\frac{1}{2}$ which implies $ \gamma(\rho_W(\frac{1}{2}))=N^p(\frac{1}{2})$ using duality between compatibility of noisy spin measurements and steerability of Werner state as discussed in section~\ref{appendixF}.
\end{proof}
Similarly one could show that 
\begin{corollary}
  $R_{\text{POVM}}(4)=R(4)=\frac{1}{3}\Leftrightarrow  \gamma(\rho_W(\rho_W(\frac{1}{3}))=N(\frac{1}{3})=4$
\label{lem: n=4}
\end{corollary}
\begin{proposition}
$R_{\text{POVM}}^p(n)=R^p(n)\Leftrightarrow \gamma^p(\omega_r)=N^p(r)$.
\label{thm:planar_SM}
\end{proposition}
\begin{proof}
\yujie Without loss of generality, we always assume $\{\Pi_{\lambda}\}_{\lambda=1}^n$ is rank-1 and with pairwise distinct effects. Consequently, any three of its effects are linearly independent. Moreover, since these effect are planar effect (i.e. with Bloch vector on the xz plane), any four of them must be linear dependent.   \blk \par

Suppose each effect of the spin measurement (an instance of 2-outcome measurement) $\{M_a^{(r)},\;\mathbb I-M_a^{(r)}\}$ admits a simulation
\[
M_a^{(r)}=\sum_{\lambda=1}^n p_{a|\lambda}\,\Pi_\lambda,\qquad p_{a|\lambda}\in[0,1],
\]
with no assumption that $\sum_a p_{a|\lambda}=1$. Nevertheless,  since $
\sum_{\lambda=1}^n\Pi_\lambda=\sum_{a=1}^3 M_a^{(r)}=\mathbb I,$
we always have the operator identity
\[
\sum_{\lambda=1}^n\bigl(s_\lambda-1\bigr)\Pi_\lambda=0,
\qquad \text{where}\quad s_\lambda:=\sum_{a=1}^3 p_{a|\lambda}.
\]
Thus it suffices to adjust the coefficients to obtain a single column-stochastic post-processing $p_{a|\lambda}$ with $\sum_a p_{a|\lambda}=1$ for all $\lambda$ while keeping each $M_a^{(r)}$ unchanged. (Here $\lambda$ indexes columns, and $a$ indexes raw.)

\noindent\textbf{Case $n=3$.}
By linear independence of $\{\Pi_\lambda\}_{\lambda=1}^3$ the equality 
$\sum_{\lambda}(s_\lambda-1)\Pi_\lambda=0$ forces $s_\lambda=1$ for each $\lambda$. Hence $p_{a|\lambda}$ is already normalized (column-stochastic)

\noindent\textbf{Case $n=4$.}
There must exist a non-trivial linear dependent relation
\begin{align}
\sum_{\lambda=1}^4 q_\lambda\,\Pi_\lambda=0
\quad\text{with}\quad q_\lambda:=s_\lambda-1=\sum_{a}p_{a|\lambda}-1.\label{eq:lp-r}    
\end{align}

Since $\sum_\lambda q_\lambda\,\mathrm{tr}(\Pi_\lambda)=0$ and $\mathrm{tr}(\Pi_\lambda)>0$, both $\{q_\lambda>0\}$ and $\{q_\lambda<0\}$ are nonempty. Moreover, each set should have exactly two elements, otherwise, one effect equals the mixture of the other three and can not be rank-1. Without loss of generality, we assume $q_{\lambda}>0$ (or equivalently $s_{\lambda}>1$) for $\lambda\in \{1,2\}$.

Because there are only three raw $a\in\{1,2,3\}$, the condition $s_{\lambda}=\sum_a p_{a|\lambda}>1$ for $\lambda\in \{1,2\}$ implies the existence of a raw $a$ such that $p_{a|1},p_{a|2}>0$ (as $0\le p_{a|\lambda}\le 1$).

Now, we can apply the linear dependent relation in Eq~
\ref{eq:lp-r} to shift the element $p_{a|\lambda}$ in row $a$ by:
$$
p'_{a|\lambda}(\beta):=p_{a|\lambda}-\beta\,q_\lambda,\qquad \beta\ge 0
$$
to drive $s_{\lambda}':=\sum_a p'_{a|\lambda}(\beta)$ for $\lambda\in \{1,2\}$ towards 1. 

\par Because $\sum_\lambda q_\lambda\Pi_\lambda=0$, we have that
$\sum_\lambda p'_{a|\lambda}(\beta)\Pi_\lambda=\sum_\lambda p_{a|\lambda}\Pi_\lambda=M_a^{(r)}$,
so the simulation of effect $M_a^{(r)}$ is preserved. Let $\beta^\star$ be the largest value for which $p'_{a|\lambda}(\beta)\ge 0$ for all $\lambda$. At $\beta^\star$ one of the following occurs:

(i) For at least one $\lambda$ with $q_\lambda>0$, the column sum hits $1$:
$s'_\lambda(\beta^\star)=s_\lambda-\beta^\star q_\lambda=1$; or

(ii) Some entry in row $a$ hits zero: $p'_{a|\lambda}(\beta^\star)=0$ for some $\lambda$.

In case (i) it strictly fixes one column ($s_\lambda=1$) and problem reduce to the case with $n=3$, thus the column stochastic condition is met.  \par 

In case (ii) we have reduced the support of row $a$ while preserving all effects. Now, pick another row $a'$ that  $p_{a'|\lambda}>0$ for $\lambda\in \{1,2\}$ (such a row must exist, otherwise $s_1$ and $s_2$ could not exceed 1), and repeat the same update with row $a'$. 

Each iteration either fixes a column or strictly reduces a row’s support, so the process terminates only with $s_\lambda=1$ for all $\lambda$. The final coefficients $p_{a|\lambda}$ are column-stochastic and still reproduce each $M_a^{(r)}$.
\\
\yujie 
\noindent\textbf{Case $n>5$.}: Assuming $\sum_a p_{a|\lambda}\ne 1$ for all $\lambda$. We can divide these effects in two indexes set $\mc{J}$ and $\mc{K}$ based on the sign of $q_\lambda$, and thus there must exist two pairs of $\Pi_\lambda$ such that
$$\sum_\lambda q_\lambda \Pi_\lambda =\sum_{j\in \mc J}q_j\Pi_j-\sum_{k\in \mc K}q_k\Pi_k= t_{j_1}\Pi_{j_1}^p+t_{j_2}\Pi_{j_2}^p-t_{k_1}\Pi_{k_1}^p-t_{k_2}\Pi_{k_2}^p=0,$$
with $t_{\nu}>0$. This is possible since $\{\Pi_\lambda\}$ is assumed to be a rank-1 POVM with distinct effects, and we explained this in more detail in the lemma below. 

With this linear dependent relation, we can follow the same procedure as in \textbf{$n=4$} iteratively and smooth $\sum_a p_{a|i}= 1$ one by one until we get $n=3$ 
\blk 

\end{proof}
\begin{lemma}Let $\{\Pi_\lambda=\alpha_\lambda(\mathbb{I}+\hat n_\lambda\!\cdot\!\vec\sigma)\}$ be rank-1 qubit effects whose Bloch
directions $\{\hat n_\lambda\}$ are coplanar and pairwise distinct. 
Let $\mc J,\mc K$ be disjoint triples of indices and suppose there exist 
strictly positive weights $q_\lambda$ such that
\begin{equation}
\sum_{j\in\mc J} q_j\,\Pi_j \;=\; \sum_{k\in\mc K} q_k\,\Pi_k.
\label{eq:balance6}
\end{equation}
Then there exist distinct $j_1\neq j_2\in\mc J$, $k_1\neq k_2\in\mc K$ and
$t_{j_1},t_{j_2},t_{k_1},t_{k_2}>0$ with
\begin{equation}
t_{j_1}\,\Pi_{j_1}+t_{j_2}\,\Pi_{j_2} \;=\; t_{k_1}\,\Pi_{k_1}+t_{k_2}\,\Pi_{k_2},
\label{eq:2v2}
\end{equation}
\end{lemma}
\yujie 
\begin{proof}
From \eqref{eq:balance6}, equality of identity and Bloch parts gives
\begin{equation}
\sum_{j\in\mc J} q_j\alpha_j=\sum_{k\in\mc K} q_k\alpha_k
\quad\text{and}\quad
\sum_{j\in\mc J} q_j\alpha_j \hat n_j=\sum_{k\in\mc K} q_k\alpha_k \hat n_k \;=:\; \vec v.
\end{equation}
Thus $\frac{\vec v}{\sum_{j\in\mc J} q_j\alpha_j}\in\mathrm{conv}\{\hat n_j:j\in\mc J\}\cap\mathrm{conv}\{\hat n_k:k\in\mc K\}$.
Each of these convex hulls is a polygon inscribed in the same Bloch circle. Since all
vertices lie on the circle and the vertex sets are disjoint, no vertex of one polygon
is contained in the boundary or (open) interior of the other. Therefore every vertex
of the intersection $\mathrm{conv}\{\hat n_j:j\in\mc J\}\cap\mathrm{conv}\{\hat n_k:k\in\mc K\}$
is an intersection of an edge from the first polygon with an edge from the second.
Pick such a vertex and denote it by $\vec w$. Then there exist distinct
$j_1\neq j_2\in\mc J$, $k_1\neq k_2\in\mc K$ and parameters $t,s\in(0,1)$ with
\begin{equation}
\vec w \;=\; t\hat n_{j_1}+(1-t)\hat n_{j_2} \;=\; s\hat n_{k_1}+(1-s)\hat n_{k_2}.
\end{equation}
Setting \(t_{j_1}=\tfrac{t}{\alpha_{j_1}},\; t_{j_2}=\tfrac{1-t}{\alpha_{j_2}},\;
t_{k_1}=\tfrac{s}{\alpha_{k_1}},\; t_{k_2}=\tfrac{1-s}{\alpha_{k_2}}\) (all $>0$),
both sides of \eqref{eq:2v2} equal \(\mathbb{I}+\vec w\!\cdot\!\vec\sigma\).
\end{proof}
\blk
\end{appendices}
\end{document}